\pdfoutput=1
\documentclass[american,a4paper]{dmtcs}
\usepackage[utf8]{inputenc}
\usepackage[T1]{fontenc}
\usepackage{babel}
\input{ushyphex.tex} 

\date{\small\today}

\usepackage[round]{natbib}
\usepackage{xspace}
\usepackage{enumitem}
\usepackage{array}
\usepackage{booktabs}
\usepackage{multicol}
\usepackage{multirow}
\usepackage[referable]{threeparttablex}
\usepackage{colortbl}
\usepackage[tight,TABBOTCAP]{subfigure}

\usepackage{wref}

\usepackage{xfrac}
\usepackage{amsmath}
\usepackage{amssymb}
\usepackage{dsfont}
\usepackage{bm}
\usepackage{mathtools}
\usepackage{stmaryrd}
\usepackage{tikz}
\usepackage{colonequals}
\usepackage{microtype}
\usepackage
	{url}
\usepackage{ifthen}
\usepackage{placeins}
\usepackage{needspace}

\usetikzlibrary{positioning,decorations.pathreplacing,external,calc}
\usepackage{pgfplots}
\pgfplotsset{compat=1.5}

\pgfdeclarelayer{background}
\pgfsetlayers{background,main}

\usepackage{clrscode3e}

\def\EndFor{\End\li\kw{end for} }
\def\EndIf{\End\li\kw{end if} }
\def\EndWhile{\End\li\kw{end while} }

\usepackage{float}
\floatstyle{ruled}
\newfloat{algorithm}{tbp}{loa}
\floatname{algorithm}{Algorithm}

\hypersetup{
	pdftitle={Pivot Sampling in Dual Pivot Quicksort - Exploiting
		Asymmetries in Yaroslavskiy's Partitioning Scheme}, 
	pdfauthor={Markus E. Nebel and Sebastian Wild}, 
	pdfsubject={},
	pdfkeywords={Quicksort, dual-pivot, Yaroslavskiy's partitioning method,
	median of three, average case analysis}%
}

\usepackage[amsmath,hyperref,thmmarks]{ntheorem}

\theorembodyfont{\slshape}
\theoremseparator{:}
\makeatletter
\newtheoremstyle{proofstyle}%
  {\item[\theorem@headerfont\hskip\labelsep ##1\theorem@separator]}%
  {\item[\theorem@headerfont\hskip\labelsep ##1 of ##3\theorem@separator]}
\makeatother

\newtheorem{theorem}{Theorem}[section]

\theoremstyle{plain}
\newtheorem{proposition}[theorem]{Proposition}
\newtheorem{lemma}[theorem]{Lemma}

\theoremstyle{plain}
\theorembodyfont{\upshape}

\theoremsymbol{\raisebox{-.25ex}{$\Box$}}
\qedsymbol{\raisebox{-.25ex}{$\Box$}}

\theoremstyle{proofstyle}
\renewtheorem{proof}{Proof}

\vfuzz \hfuzz	

\allowdisplaybreaks[3]

\numberwithin{equation}{section}

\newcommand\weakemph[1]{\textsl{#1}}
 
\newdimen\makeboxdimen
\newcommand\makeboxlike[3][l]{%
\setbox0=\hbox{#2}%
\global\makeboxdimen=\wd0%
\setbox1=\hbox{\makebox[\makeboxdimen][#1]{%
\makebox[0pt][#1]{#3}%
}}%
\ht1=\ht0%
\dp1=\dp0%
\box1%
}

\newcommand\plaincenter[1]{%
	\mbox{}\hfill#1\hfill\mbox{}%
}

\newcounter{inlineenum}

\newcommand*\ie{\mbox{\textit{i.\hspace{.2ex}e.}}}
\newcommand*\eg{\mbox{\textit{e.\hspace{.2ex}g.}}}

\newcommand\E{\mathop{\mbox{$\mathbb{E}$}}\nolimits}
\newcommand\given{\;|\;}

\newcommand\R{\reals}
\newcommand\N{\naturals}

\newcommand\Oh{O}
\def\.{\mskip1mu}

\newcommand{\Prob}{\ensuremath{\mathbb{P}}}

\newcommand\ui[2]{#1^{\smash{(}#2\smash{)}}}

\newcommand{\vect}[1]{\boldsymbol{\mathbf{#1}}}

\newcommand\eqdist{	
	\mathchoice{
		\mathrel{\overset{\raisebox{0ex}{$\scriptstyle \cal D$}}=}%
	}{
		\mathrel{\like{=}{%
			\overset{\raisebox{-1ex}{$\scriptscriptstyle \cal D$}}=%
		}}%
	}{
		\mathrel{\overset{\cal D}=}%
	}{
		\mathrel{\overset{\cal D}=}%
	}%
}
\newcommand\eqdistt{\mathrel{\smash{\stackrel{\scriptscriptstyle\mathcal D}=}}}

\newcommand\ppe{\phantom{=}}

\newcommand\like[3][c]{\makeboxlike[#1]{\ensuremath{#2}}{\ensuremath{#3}}}

\newcommand\uniform{\mathcal U}
\newcommand\bernoulli{\mathrm B}
\newcommand\hypergeometric{\mathrm{HypG}}
\newcommand\multinomial{\mathrm{Mult}}
\newcommand\binomial{\mathrm{Bin}}
\newcommand\dirichlet{\mathrm{Dir}}

\newcommand\BetaFun{\mathrm B}

\renewcommand\given{\mathbin{\mid}}
\newcommand\harm[1]{\ensuremath{\mathcal{H}_{#1}}}

\newcommand\ce{\colonequals}

\newcommand\rel[1]{\mathrel{\:{#1}\:}}
\newcommand\wrel[1]{\mathrel{\;{#1}\;}}
\newcommand\wwrel[1]{\mathrel{\;\;{#1}\;\;}}
\newcommand\bin[1]{\mathbin{\:{#1}\:}}
\newcommand\wbin[1]{\mathbin{\;{#1}\;}}

\newcommand{\eqwithref}[2][c]{%
	\relwithref[#1]{#2}{=}%
}
\newcommand{\relwithref}[3][c]{%
	\mathrel{\underset{\mathclap{\makebox[\widthof{$=$}][#1]{\scriptsize\wref{#2}}}}{#3}}%
}

\newcommand\toll[2][]{%
	\ensuremath{%
	\ifthenelse{\equal{#1}{}}{%
		T_{\!#2}%
	}{%
		T_{\!#2}({#1})%
	}}%
}

\newcommand\istoll[2][]{%
	\ensuremath{%
	\ifthenelse{\equal{#1}{}}{%
		\iscost_{\!#2}%
	}{%
		\iscost_{\!#2}({#1})%
	}}%
}

\newcommand\insertsortcost{W}
\newcommand\iscost{\insertsortcost}

\newcommand\bytecodes{\mathit{BC}}

\newcommand\values[1]{#1}

\newcommand\positionsets[1]{\mathcal{#1}}

\newcommand\numberat[2]{\values{#1}\mbox{\emph{@}}\.\positionsets{#2}}
\newcommand\satK{\numberat sK}
\newcommand\latK{\numberat lK}
\newcommand\satG{\numberat sG}

\newcommand\indicator[1]{\mathds 1_{\{#1\}}}

\newcommand\discreteEntropy[1][\vect t]{\ensuremath{\mathrm{H}(#1)}\xspace}
\newcommand\contentropy[1][\vect\tau]{%
	\ensuremath{\mathchoice{
		{\mathrm{H^*}} \ifthenelse{\equal{#1}{}}{ }{ (#1) }
	}{
		{\mathrm{H^*}} \ifthenelse{\equal{#1}{}}{ }{ (#1) }
	}{
		{\mathrm{H}}^* \ifthenelse{\equal{#1}{}}{ }{ (#1) }
	}{
		{\mathrm{H}}^* \ifthenelse{\equal{#1}{}}{ }{ (#1) }
	}}\xspace%
}

\newcommand\arrayA{%
	\ensuremath{\mathchoice{
		\smash{\raisebox{-.2pt}{\scalebox{1.25}[1.18]{$\mathtt{A}$}}}%
	}{
		\smash{\raisebox{-.2pt}{\scalebox{1.25}[1.18]{$\mathtt{A}$}}}%
	}{
		\smash{\raisebox{-.2pt}{\scalebox{1.25}[1.18]{$\scriptstyle\mathtt{A}$}}}%
	}{
		\smash{\raisebox{-.2pt}{\scalebox{1.25}[1.18]{$\scriptscriptstyle\mathtt{A}$}}}%
	}}\xspace%
}

\usepackage{manfnt}

\newcommand*\generalYaros[2]{%
	\ensuremath{Y_{#1}^{#2}}\xspace%
}
\newcommand*\generalYarostM{%
	\generalYaros{\mkern-1mu \vect t}{\isthreshold}%
}
\newcommand*\isthreshold{\ensuremath{\mathnormal{w}}\xspace}

\let\oldparagraph\paragraph
\makeatletter%
\renewcommand\paragraph{%
    \@ifstar{\myparagraphStar}{\myparagraphNoStar}%
}
\makeatother%
\newcommand\myparagraphStar[1]{%
	\oldparagraph*{#1.}%
}
\newcommand\myparagraphNoStar[2][]{%
	\ifthenelse{\equal{#1}{}}{%
		\oldparagraph[#2]{#2.}%
	}{%
		\oldparagraph[#1]{#2.}%
	}%
}

\colorlet{symmetriccolor}{black!10}

\hypersetup{
 	 breaklinks=true,
 	 pdfborder={0 0 0}, 
 	 colorlinks=false   
}

\makeatletter

\renewcommand\l@section[2]{\@dottedtocline{1}{1.5em}{2.3em}{\textbf{#1}}{\textbf{#2}}}
\renewcommand\l@section[2]{\@dottedtocline{1}{1.5em}{2.3em}{#1}{#2}}
\makeatother

\let\notag\relax
\let\nonumber\relax

\title[Pivot Sampling in Dual-Pivot Quicksort]{%
	Pivot Sampling in Dual-Pivot Quicksort\\%
	{\Large Exploiting Asymmetries in Yaroslavskiy's Partitioning Scheme} 
	}

\author{%
	Markus E.\ Nebel\addressmark{1}\addressmark{2}%
	\thanks{The order of authors follows the Hardy-Littlewood rule, \ie, 
	it is alphabetical by last name.}%
	\and Sebastian Wild\addressmark{1}
}

\address{%
	\addressmark{1}%
	Computer Science Department, University of Kaiserslautern\\
	\addressmark{2}%
	Department of Mathematics and Computer Science, University of Southern Denmark
}


\begin{document}

\maketitle

\begin{abstract}
	\noindent\textsf{\textbf{Abstract:\quad}}\small
	The new dual-pivot Quicksort by Vladimir Yaroslavskiy\,---\,used in Oracle's
	Java runtime library since version 7\,---\,features intriguing asymmetries in
	its behavior. They were shown to cause a basic variant of this algorithm to use
	less comparisons than classic single-pivot Quicksort implementations.
	In this paper, we extend the analysis to the case where the two pivots are
	chosen as fixed order statistics of a random sample and give the precise
	leading term of the average number of comparisons, swaps and executed Java
	Bytecode instructions. 
	It turns out that\,---\,unlike for classic Quicksort, where it is optimal to
	choose the pivot as median of the sample\,---\,the asymmetries in
	Yaroslavskiy's algorithm render pivots with a systematic skew more
	efficient than the symmetric choice.
	Moreover, the optimal skew heavily depends on the employed cost measure;
	most strikingly, abstract costs like the number of swaps and comparisons yield
	a very different result than counting Java Bytecode instructions,
	which can be assumed most closely related to actual running time. 
\end{abstract}

\section{Introduction}

Quicksort is one of the most efficient comparison-based sorting algorithms and
is thus widely used in practice, for example in the sort implementations
of the C++ standard library and Oracle's Java runtime library.
Almost all practical implementations are based on the highly tuned version of
\citet{Bentley1993}, often equipped with the strategy of
\citet{Musser1997} to avoid quadratic worst case behavior.
The Java runtime environment was no exception to this\,---\,up to version~6.
With version~7 released in~2009, however, Oracle broke with this tradition and
replaced its tried and tested implementation by a dual-pivot Quicksort
with a new partitioning method proposed by Vladimir Yaroslavskiy.

The decision was based on extensive running time experiments that clearly
favored the new algorithm.
This was particularly remarkable as earlier analyzed dual-pivot 
variants had not shown any potential for performance gains over classic
single-pivot Quicksort \citep{Sedgewick1975,hennequin1991analyse}.
However, we could show for pivots from fixed array positions (\ie\ no
sampling) that Yaroslavskiy's asymmetric partitioning method
beats classic Quicksort in the comparison model:
asymptotically $1.9\,n\ln n$ vs.\ $2\,n\ln n$ comparisons on
average \citep{Wild2012}.
As these savings are opposed by a large increase in the number of swaps, 
the overall competition still remained open. 
To settle it, we compared two Java implementations of the Quicksort
variants and found that Yaroslavskiy's method actually executes \emph{more} Java
Bytecode instructions on average \citep{Wild2013Quicksortarxiv}.
A possible explanation why it still shows better running times was recently
given by \citet{Kushagra2014}:
Yaroslavskiy's algorithm needs fewer scans over
the array than classic Quicksort, and is thus more efficient in the
\textsl{external memory model}.

Our analyses cited above ignore a very effective strategy in Quicksort:
for decades, practical implementations choose their pivots as \emph{median of a
random sample} of the input to be more efficient (both in terms of average
performance and in making worst cases less likely).
Oracle's Java~7 implementation also employs this optimization: 
it chooses its two pivots as the \emph{tertiles of five} sample elements.
This equidistant choice is a plausible generalization, since
selecting the pivot as median is known to be optimal for classic Quicksort
\citep{Sedgewick1975,Martinez2001}.

However, the classic partitioning methods treat elements smaller
and larger than the pivot in symmetric ways\,---\,unlike Yaroslavskiy's
partitioning algorithm:
depending on how elements relate to the two pivots, one of
\emph{five} different execution paths is taken in the partitioning loop, and
these can have highly different costs!
How often each of these five paths is taken depends on the \emph{ranks} of the
two pivots, which we can push in a certain direction by selecting \emph{other} order
statistics of a sample than the tertiles.
The partitioning costs alone are then minimized if the cheapest execution path
is taken all the time. 
This however leads to very unbalanced distributions of sizes
for the recursive calls, such that a \emph{trade-off} between partitioning costs
and balance of subproblem sizes results.

We have demonstrated experimentally that there is potential to tune dual-pivot
Quicksort using skewed pivots \citep{Wild2013Alenex}, but only considered a
small part of the parameter space.
\textbf{%
It will be the purpose of this paper to identify the optimal way to
sample pivots by means of a precise analysis of the resulting overall costs,
}
and to validate (and extend) the empirical findings that way.

\paragraph{Related work}

Single-pivot Quicksort with pivot sampling has been intensively studied over the
last decades
\citep{VanEmden1970,Sedgewick1975,Sedgewick1977,hennequin1991analyse,%
	Martinez2001,neininger2001multivariate,Chern2001a,Durand2003pseudonine}.
We heavily profit from the mathematical foundations laid by these authors.
There are scenarios where, even for the symmetric, classic Quicksort, a
skewed pivot can yield benefits over median of~$k$ \citep{Martinez2001,kaligosi2006branch}. 
An important difference to Yaroslavskiy's algorithm is, however, that the
situation remains symmetric:
a relative pivot rank $\alpha < \frac12$ has the same effect as
one with rank~$1-\alpha$.
For dual-pivot Quicksort with an \emph{arbitrary} partitioning method,
\citet{Aumuller2013icalp} establish a lower bound of asymptotically 
$1.8 \, n\ln n$ comparisons
and they also propose a partitioning method that attains this bound.

\paragraph{Outline}
After listing some general notation,
\wref{sec:generalized-yaroslavskiy-quicksort} introduces the subject of study.
\wref{sec:results} collects the main analytical results of this paper, whose 
proof is divided into 
\wref[Sections]{sec:distributional-analysis}, \ref{sec:expectations}
and~\ref{sec:solution-recurrence}.
Arguments in the main text are kept concise, but the interested reader is
provided with details in the appendix. 
The algorithmic consequences of our analysis are discussed in
\wref{sec:asymmetries-everywhere}. 
\wref{sec:conclusion} concludes the paper.

\section{Notation and Preliminaries}
\label{sec:notation}
\enlargethispage{\baselineskip}

We write vectors in bold font, for example 
$\vect t=(t_1,t_2,t_3)$.
For concise notation, we use expressions like $\vect t + 1$ to mean
\emph{element-wise} application, \ie, $\vect t + 1 = (t_1+1,t_2+1,t_3+1)$.
By $\dirichlet(\vect\alpha)$, we denote a random variable with \textsl{Dirichlet
distribution} and shape parameter
$\vect\alpha = (\alpha_1,\ldots,\alpha_d) \in \R_{>0}^d$.
Likewise for parameters $n\in\N$ and $\vect p =
(p_1,\ldots,p_d) \in [0,1]^d$ with $p_1+\cdots+p_d=1$, we write
$\multinomial(n,\vect p)$ for a random variable with \textsl{multinomial
distribution} with $n$ trials.
$\hypergeometric(k,r,n)$ is a random variable with
\textsl{hypergeometric distribution}, \ie, the number of red balls when
drawing $k$ times without replacement from an urn of $n\in\N$ balls,
$r$ of which are red, (where $k,r\in\{1,\ldots,n\}$).
Finally, $\uniform(a,b)$ is a random variable uniformly distributed in the
interval $(a,b)$, and $\bernoulli(p)$ is a Bernoulli variable with probability
$p$ to be $1$.
We use ``$\eqdist$'' to denote equality in distribution.

As usual for the average case analysis of sorting algorithms, we assume the
\textsl{random permutation model}, \ie, all elements are
different and every ordering of them is equally likely.
The input is given as array $\arrayA$ of length $n$ and we denote the initial
entries of $\arrayA$ by $U_1,\ldots,U_n$.
We further assume that $U_1,\ldots,U_n$ are
i.\,i.\,d.\ uniformly $\uniform(0,1)$ distributed;
as their ordering forms a random permutation \citep{mahmoud2000sorting}, this
assumption is without loss of generality. 
Some further notation specific to our analysis is introduced below; for
reference, we summarize all notations used in this paper in \wref{app:notations}.

\section{Generalized Yaroslavskiy Quicksort}
\label{sec:generalized-yaroslavskiy-quicksort}

In this section, we review Yaroslavskiy's partitioning method and combine it
with the pivot sampling optimization to obtain what we call the
\textsl{Generalized Yaroslavskiy Quicksort} algorithm.
We leave some parts of the algorithm unspecified here,
but give a full-detail implementation in the appendix. 
The reason is that \emph{preservation of randomness} is somewhat tricky to
achieve in presence of pivot sampling, but vital for precise analysis.
The casual reader might content him- or herself with our promise that
everything turns out alright in the end; the interested reader is invited
to follow our discussion of this issue in \wref{app:algorithms}.

\subsection{Generalized Pivot Sampling}
\label{sec:general-pivot-sampling}

Our pivot selection process is declaratively specified as
follows, where $\vect t = (t_1,t_2,t_3) \in \N^3$ is a fixed parameter:
choose a random sample $\vect V = (V_1,\ldots,V_k)$ of size 
$k = k(\vect t) \ce t_1+t_2+t_3+2$ from the elements
and denote by
$(V_{(1)},\ldots,V_{(k)})$ the \emph{sorted}\hspace{1pt}%
\footnote{%
	In case of equal elements any possible ordering will do. 
	However in this paper, we assume distinct elements. 
}
sample,
\ie,
$
	V_{(1)} \le V_{(2)} \le \cdots \le V_{(k)}
$.
Then choose the two pivots $P \ce V_{(t_1+1)}$ and $Q\ce V_{(t_1+t_2+2)}$ such
that they divide the sorted sample into three regions of respective sizes $t_1$,
$t_2$ and~$t_3$:
\begin{multline}
	\underbrace{V_{(1)} \ldots V_{(t_{1})}}
			_{t_{1}\,\mathrm{elements}}
	\wrel\le
	\underbrace{V_{(t_{1}+1)}} _ {=P}
	\wrel\le
	\underbrace{V_{(t_{1}+2)} \ldots V_{(t_{1}+t_{2}+1)}}
			_{t_{2}\,\mathrm{elements}}
	\wrel\le
	\underbrace{V_{(t_{1}+t_{2}+2)}} _ {=Q}
	\wrel\le
	\underbrace{V_{(t_{1}+t_{2}+3)} \ldots V_{(k)}}
			_{t_{3}\,\mathrm{elements}}
	\;.
\end{multline}
Note that by definition, $P$ is the small(er) pivot and $Q$ is the
large(r) one.
We refer to the $k-2$ elements of the sample that are not chosen as
pivot as \emph{``sampled-out''};
$P$~and $Q$ are the chosen \emph{pivots}. 
All other elements\,---\,those which have not been part of the sample\,---\,are
referred to as \emph{ordinary} elements. 
Pivots and ordinary elements together form the set of \emph{partitioning
elements},
(because we exclude sampled-out elements from partitioning).

\subsection{Yaroslavskiy's Dual Partitioning Method}
\label{sec:yaroslavskiys-partitioning-method}
\enlargethispage{2\baselineskip}

In bird's-eye view, Yaroslavskiy's partitioning method consists of two
indices, $k$ and $g$, that start at the left resp.\ right end of \arrayA and
scan the array until they meet. 
Elements left of $k$ are smaller or equal than $Q$, elements right of $g$ are
larger.
Additionally, a third index $\ell$ lags behind $k$ and separates elements
smaller than $P$ from those between both
pivots.
Graphically speaking, the invariant of the algorithm is as follows:

\smallskip
\begin{quote}
	\mbox{}\hfill%
	\begin{tikzpicture}[
		yscale=0.5, xscale=0.6,
		baseline=(ref.south),
		every node/.style={font={}},
		semithick,
	]	
	
	\draw (-.75,0) -- ++(14.5,0) -- ++(0,1) -- ++(-14.5,0) -- cycle;
	\node at (-.375, .5) {$P$} ;
	\node at (13.375, .5) {$Q$} ;
	\draw (0,0) -- ++(0,1);
	\draw (13,0) -- ++(0,1);
	
	\node at (1.5,0.5) {$< P$};
	\draw (3,1) -- ++ (0,-1);
	\node at (3.3,-0.4) {$\ell$};
	
	\node at (12,0.5) {$\ge Q$};
	\draw (11,1) -- ++ (0,-1);
	\node at (10.7,-0.4) {$g$};
	
	\node at (5,0.5) {$P\le \circ\le Q$};
	\draw (7,1) -- ++(0,-1);
	\node at (7.3,-0.4) {$k$};

	\node[below] at (10.7,-0.6) {$\leftarrow$};
	\node[below] at (3.3,-0.6) {$\rightarrow$};
	\node[below] at (7.3,-0.6) {$\rightarrow$};
	
	\node[inner sep=0pt] (ref) at (9,0.5) {?};
	\end{tikzpicture}%
	\hfill\mbox{}
\end{quote}

\noindent
We write~$\positionsets{K}$ and~$\positionsets{G}$ for the sets of all indices
that $k$ resp.~$g$ attain in the course of the partitioning process.
Moreover, we call an element \emph{small}, \emph{medium}, or \emph{large} if
it is smaller than~$P$, between $P$ and $Q$, or larger than~$Q$, respectively.
The following properties of the algorithm are needed for the analysis, 
(see \citet{Wild2012,Wild2013Quicksortarxiv} for details):
\begin{enumerate}[label=(Y\arabic*)]
	\item \label{prop:in-K-first-with-p} 
		Elements $U_i$, $i \in \positionsets{K}$, are first compared with $P$.
		Only if $U_i$ is not small, it is also compared to $Q$.
	\item \label{prop:in-G-first-with-q} 
		Elements $U_i$, $i \in \positionsets{G}$, are first compared with~$Q$.
		If they are not large, they are also compared to~$P$.
	\item \label{prop:small-eventually-swappel-left} 
		Every small element eventually causes one swap to put it behind $\ell$.
	\item \label{prop:latK-smatG-swapped-in-pairs}
		The large elements located in $\positionsets{K}$ and the non-large
		elements in $\positionsets{G}$ are always swapped in pairs.
\end{enumerate}
For the number of comparisons we will thus need to
count the large elements $U_i$ with $i \in \positionsets{K}$; we abbreviate
their number by ``$\latK$''.
Similarly, $\satK$ and $\satG$ denote the number of small elements in $k$'s
resp.\ $g$'s range.

When partitioning is finished, $k$ and $g$ have met and thus $\ell$ and $g$
divide the array into three ranges, containing the small, medium resp.\
large (ordinary) elements, which are then sorted recursively.
For subarrays with at most \isthreshold elements, we switch to Insertionsort,
(where \isthreshold is constant and at least $k$).
The resulting algorithm, Generalized Yaroslavskiy Quicksort with pivot
sampling parameter $\vect t=(t_1,t_2,t_3)$ and Insertionsort threshold
\isthreshold, is henceforth called $\generalYarostM$.

\section{Results}
\label{sec:results}

For $\vect t \in \N^3$ and $\harm{n}$ the $n$th harmonic number, we define the
\textsl{discrete entropy} $\discreteEntropy$ of $\vect t$ as
\begin{align}
\label{eq:discrete-entropy}
		\discreteEntropy[\vect t]
	&\wwrel=
		\sum_{l=1}^3 \frac{t_l+1}{k+1} 
				(\harm{k+1} - \harm{t_l+1})
	\;. 
\end{align} 
The name is justified by the following connection between \discreteEntropy and
the \weakemph{entropy function} $\contentropy[]$ of information theory: 
for the sake of analysis, let $k\to\infty$, such that  
ratios ${t_l}/k$ converge to constants~$\tau_l$. Then
\begin{align}
\label{eq:limit-g-entropy}
		\discreteEntropy
	&\wwrel\sim
		- \sum_{l=1}^3 \tau_l \bigl( \ln(t_l+1) - \ln(k+1) \bigr)
	\wwrel\sim	  - \sum_{l=1}^3 \tau_l \ln(\tau_l)
	\wwrel{\equalscolon} \contentropy
	\;.
\end{align}
The first step follows from the asymptotic equivalence
$\harm{n} \sim \ln(n)$ as $n\to\infty$. 
\wref{eq:limit-g-entropy} shows that for large~$\vect t$, the maximum of
$\discreteEntropy$ is attained for 
$\tau_1 = \tau_2 = \tau_3 = \frac13$.
Now we state our main result:

\begin{theorem}[Main theorem]
\label{thm:expected-costs}
	Generalized Yaroslavskiy Quicksort with pivot sampling
	parameter $\vect t = (t_1,t_2,t_3)$ performs on average
	$ C_n \sim \frac{a_C}{\discreteEntropy} \, n\ln n$
	comparisons and
	$ S_n \sim \frac{a_S}{\discreteEntropy} \, n\ln n$
	swaps to sort a random permutation of $n$ elements, where 
	\begin{align*}
			a_C
		&\wrel=
				  1 + \frac{t_2+1}{k+1}
				    + \frac{(2t_1 + t_2 + 3)(t_3+1)}{(k+1)(k+2)}
	\quad\text{and}\quad
			a_S
		\wrel=
				  \frac{t_1+1}{k+1}
				+ \frac{(t_1 + t_2 + 2)(t_3 + 1)}{(k+1)(k+2)} \,.
	\end{align*}
	Moreover, if the partitioning loop is implemented as in Appendix~C 
	of~\citep{Wild2013Quicksortarxiv}, it executes on average 
	$\bytecodes_{\!n} \sim \frac{a_{\bytecodes}}{\discreteEntropy} \, n\ln n$
	Java Bytecode instructions to sort a random permutation of size $n$ with
	\begin{align*}
			a_\bytecodes
		&\wwrel=
				  		10
				\bin+ 	13 \frac{t_1+1}{k+1}
				\bin+  	 5 \frac{t_2+1}{k+1}
				\bin+ 	11 \frac{(t_1+t_2 + 2)(t_3+1)}{(k+1)(k+2)}
				\bin+ 	   \frac{(t_1+1)(t_1 + t_2 + 3)}{(k+1)(k+2)}
			\;.
	\end{align*}
\end{theorem}

\noindent
The following sections are devoted to the proof of \wref{thm:expected-costs}.
\wref{sec:distributional-analysis} sets up a recurrence of costs
and characterizes the distribution of costs of one partitioning step.
The expected values of the latter are computed in \wref{sec:expectations}.
Finally, \wref{sec:solution-recurrence} provides a generic solution to the
recurrence of the expected costs; in combination with the expected partitioning
costs, this concludes our proof.

\section{Distributional Analysis}
\label{sec:distributional-analysis}

\subsection{Recurrence Equations of Costs}
\label{sec:recurrence-quicksort}

Let us denote by $C_n$ the \emph{costs} of \generalYarostM on a random
permutation of size~$n$\,---\,where different ``cost measures'', like the number
of comparisons, will take the place of $C_n$ later. 
$C_n$~is a non-negative \emph{random} variable whose distribution
depends on~$n$.
The total costs decompose into those for the first partitioning step plus the
costs for recursively solving subproblems.
As Yaroslavskiy's partitioning method preserves randomness (see
\wref{app:algorithms}), we can express the total costs
$C_n$ recursively in terms of the same cost function with smaller arguments:
for sizes $J_1$, $J_2$ and $J_3$ of the three subproblems, the costs of
corresponding recursive calls are distributed like $C_{J_1}$, $C_{J_2}$
and~$C_{J_3}$, and conditioned on $\vect J = (J_1,J_2,J_3)$, these
random variables are independent.
Note, however, that the subproblem sizes are
themselves random and inter-dependent.
Denoting by $T_n$ the costs of the first partitioning step, 
we obtain the following \emph{distributional recurrence} for the family
$(C_n)_{n\in\N}$ of random variables:
\begin{align}
\label{eq:distributional-recurrence}
	C_n &\wwrel\eqdist \begin{cases}
			T_n \wbin+ C_{J_1} + C'_{J_2} + C''_{J_3} ,
			& \text{for } n > \isthreshold ; \\
		\iscost_n ,
			& \text{for } n \le \isthreshold .
	\end{cases}
\end{align}
Here $\iscost_n$ denotes the cost of
{Insertionsort}\kern.5pt ing a random permutation of size $n$.
$(C'_j)_{j\in\N}$ and $(C''_j)_{j\in\N}$ are independent copies of
\smash{$(C_j)_{j\in\N}$}, \ie, for all $j$, the variables $C_j$, $C'_j$ and
$C''_j$ are identically distributed and for all \smash{$\vect j \in \N^3$}, 
$C_{j_1}$, $C'_{j_2}$ and $C''_{j_3}$ are totally independent, and they are also
independent of $T_n$. 
We call $T_n$ the \emph{toll function} of the recurrence, as it quantifies
the ``toll'' we have to pay for unfolding the recurrence once.
Different cost measures only differ in the toll functions, such that we can
treat them all in a uniform fashion by studying \wref{eq:distributional-recurrence}.
Taking expectations on both sides, we
find a recurrence equation for the \emph{expected} costs $\E[C_n]$:
\begin{align}
\label{eq:ECn-recurrence}
		\E[C_n]
	&\wwrel= \begin{cases}
		\displaystyle
		\E[T_n] 
		\wbin+ 
		\sum_{\mathclap{\substack{\vect j=(j_1,j_2,j_3) \\ j_1+j_2+j_3=n-2}}}
			\Prob(\vect J = \vect j)
			\bigl(
				\E[C_{j_1}] + \E[C_{j_2}] + \E[C_{j_3}]
			\bigr),
		& \text{for } n > \isthreshold  ; \\[5ex]
		\E[\iscost_n] ,
		& \text{for } n \le \isthreshold .
	\end{cases}
\end{align}
A simple combinatorial argument gives access to $\Prob(\vect J = \vect j)$, 
the probability of $\vect J = \vect j$:
of the $\binom nk$ different size $k$ samples of $n$
elements, those contribute to the probability of $\{\vect J =
\vect j\}$, in which exactly $t_1$ of the sample elements are chosen from 
the overall $j_1$ small elements; and likewise $t_2$ of the $j_2$ medium
elements and $t_3$ of the $j_3$ large ones are contained in the sample.
We thus have

\begin{align}
		\Prob(\vect J = \vect j)
	&\wwrel=
		\binom{j_1}{t_1} \binom{j_2}{t_2} \binom{j_3}{t_3}
		\bigg/
		\binom nk
		\;.
\label{eq:prob-for-J-equals-j}
\end{align}

\needspace{5\baselineskip}
\subsection{Distribution of Partitioning Costs}

Let us denote by $I_1$, $I_2$ and $I_3$ the number of small, medium and large
elements among the ordinary elements, (\ie, $I_1+I_2+I_3 =
n-k$)\,---\,or equivalently stated, $\vect I = (I_1,I_2,I_3)$ is the (vector
of) sizes of the three partitions (excluding sampled-out elements).
Moreover, we define the indicator variable 
$\delta = \indicator{U_\chi \rel> Q}$ to
account for an idiosyncrasy of Yaroslavskiy's algorithm (see the proof of
\wref{lem:distribution-partitioning-comparisons}), 
where $\chi$ is the point where indices $k$ and $g$ first meet.
As we will see, we can characterize the distribution of partitioning costs
\emph{conditional} on $\vect I$, \ie, when considering $\vect I$ \emph{fixed}.

\subsubsection{Comparisons}
For constant size samples, only the comparisons during the partitioning process
contribute to the linearithmic leading term of the asymptotic average costs, as
the number of partitioning steps remains linear.
We can therefore ignore comparisons needed for sorting the sample.
As \isthreshold is constant, the same is true for subproblems of size at
most \isthreshold that are sorted with Insertionsort.
It remains to count the comparisons during the first partitioning step, where
contributions that are uniformly bounded by a constant can likewise be
ignored.

\begin{lemma}
\label{lem:distribution-partitioning-comparisons}
	Conditional on the partition sizes $\vect I$, the number of comparisons
	$\toll C = \toll [n] C$ in the first partitioning step of \generalYarostM on a
	random permutation of size $n>\isthreshold$ fulfills
	\begin{align}
	\label{eq:distribution-partitioning-comparisons-exact}
			\toll [n] C
		&\wrel=
			(n-k) 
			\bin+ I_2 
			\bin+ (\latK) 
			\bin+ (\satG)
			\bin+ 2\delta
	\\	&\wrel\eqdist
	\label{eq:distribution-partitioning-comparisons-dist}
			(n-k) 
			\bin+ I_2 
			\bin+ \hypergeometric(I_1+I_2,I_3,n-k)
			\bin+ \hypergeometric(I_3,I_1,n-k)
			\bin+ 3\bernoulli\bigl(\tfrac{I_3}{n-k}\bigr)
		\;.
	\end{align} 
\end{lemma}

\begin{proof}
Every ordinary element is compared to at least one of the pivots,
which makes $n-k$ comparisons.
Additionally, for all medium elements, the second comparison is inevitably
needed to recognize them as ``medium'', and there are $I_2$ such elements.
Large elements only cause a second comparison if they are first compared
with~$P$, which happens if and only if they are located in $k$'s range,
see~\ref{prop:in-K-first-with-p}.
We abbreviated the (random) number of large elements in $\positionsets K$ as
$\latK$.
Similarly, $\satG$ counts the second comparison for all small elements found in
$g$'s range, see \ref{prop:in-G-first-with-q}.

The last summand $2\delta$ accounts for a technicality in Yaroslavskiy's
algorithm. 
If $U_\chi$, the element where $k$ and $g$ meet, is
large, then index $k$ overshoots $g$ by one, which causes two additional
(superfluous) comparisons with this element.
$\delta = \indicator{U_\chi \rel> Q}$ is the indicator
variable of this event.
This proves~\wref{eq:distribution-partitioning-comparisons-exact}.

For the equality in distribution, recall that $I_1$, $I_2$ and $I_3$ are the
number of small, medium and large elements, respectively. 
Then we need the cardinalities of $\positionsets{K}$ and $\positionsets{G}$.
Since the elements right of $g$ after partitioning are exactly all large
elements, we have
$
		|\positionsets{G}|
	=
		I_3
$ 
and
$
		|\positionsets{K}|
	=
		I_1+I_2+\delta
$;
(again, $\delta$ accounts for the overshoot, see \citet{Wild2013Quicksortarxiv}
for detailed arguments).
The distribution of $\satG$, conditional on~$\vect I$, is now given by the
following urn model:
we put all $n-k$ ordinary elements in an urn and draw their
positions in \arrayA.
$I_1$ of the elements are colored red (namely the small ones), the rest is black
(non-small).
Now we draw the $|\positionsets{G}| = I_3$ elements in $g$'s range from
the urn without replacement. 
Then $\satG$ is exactly the number of red (small)
elements drawn and thus 
$\satG \rel\eqdist \hypergeometric(I_3, I_1, n-k)$.

The arguments for $\latK$ are similar, however the additional $\delta$
in $|\positionsets{K}|$ needs special care.
As shown in the proof of Lemma~3.7 of \citet{Wild2013Quicksortarxiv}, 
the additional element in $k$'s range for the case $\delta=1$ is
$U_\chi$, which then is large by definition of~$\delta$.
It thus simply contributes as additional summand:
$\latK \rel\eqdist \hypergeometric(I_1+I_2, I_3, n-k) + \delta$.
Finally, the distribution of $\delta$ is Bernoulli
$\bernoulli\bigl(\tfrac{I_3}{n-k}\bigr)$, since conditional on $\vect I$, the
probability of an ordinary element to be large is $I_3 / (n-k)$.
\end{proof}

\subsubsection{Swaps}
As for comparisons, only the swaps in the partitioning step contribute to the
leading term asymptotics.

\begin{lemma}
\label{lem:distribution-partitioning-swaps}
	Conditional on the partition sizes $\vect I$, the number of swaps
	$\toll S = \toll [n] S$ in the first partitioning step of \generalYarostM on a
	random permutation of size $n>\isthreshold$ fulfills
	\begin{align*}
			\toll [n] S
		&\wwrel=
			 I_1 
			\bin+ (\latK) 
		\wwrel{\rel{\eqdist}}
			I_1 
			\bin+ \hypergeometric(I_1+I_2, I_3, n-k) 
			\bin+ \bernoulli\bigl(\tfrac{I_3}{n-k}\bigr)
		\;.
	\end{align*} 
\end{lemma}

\begin{proof}
No matter where a small element is located initially, it will eventually incur
one swap that puts it at its final place (for this partitioning step) to the
left of $\ell$, see \ref{prop:small-eventually-swappel-left};
this gives a contribution of $I_1$.
The remaining swaps come from the ``crossing pointer'' scheme, where $k$ stops
on the first large and $g$ on the first non-large element, which are then
exchanged in one swap \ref{prop:latK-smatG-swapped-in-pairs}.
For their contribution, it thus suffices to count the large elements in $k$'s
range, that is $\latK$.
The distribution of $\latK$ has already been discussed in the proof of
\wref{lem:distribution-partitioning-comparisons}.
\end{proof}

\subsubsection{Bytecode Instructions}

A closer investigation of the partitioning method reveals the number
of executions for every single Bytecode instruction in the algorithm.
Details are omitted here; the analysis is very similar to the case without
pivot sampling that is presented in detail in \citep{Wild2013Quicksortarxiv}.

\begin{lemma}
\label{lem:distribution-partitioning-bytecodes}
	Conditional on the partition sizes $\vect I$, the number of executed Java
	Bytecode instructions $\toll \bytecodes = \toll [n] \bytecodes$ of the first
	partitioning step of \generalYarostM{}\,---\,implemented as in Appendix~C of
	\citep{Wild2013Quicksortarxiv}\,---\,fulfills
	on a random permutation of size	$n>\isthreshold$
	\vspace{-1ex}
	\begin{small}%
	\begin{multline*}
			\toll[n]\bytecodes
		\wrel{\rel{\eqdist}}		
					10 n + 13 I_1 + 5 I_2
			+ 		11\, \hypergeometric(I_1+I_2,I_3,n-k)
			+		   \hypergeometric(I_1,I_1+I_2,n-k)
			\wbin+	\Oh(1) \;.
	\end{multline*}
	\end{small}%
\qed\end{lemma}


\noindent
\textbf{Other cost measures} can be analyzed similarly, \eg, 
the analysis of \citet{Kushagra2014} for I/Os in the \textsl{external
memory model} is easily generalized to pivot sampling. We omit it
here due to space~constraints.

\subsubsection{Distribution of Partition Sizes}

\begin{figure}
	\def\r{1.5pt}
	\plaincenter{%
		\begin{tikzpicture}[
			every node/.style={font=\footnotesize},
		]
			\useasboundingbox (0,-.2) rectangle (5,0.4) ;
			
			\draw[|-|] (0,0) node[below=.5ex] {$0$}  -- (5,0) node[below=.5ex] {$1$};
			\filldraw (1,0) circle (\r) node[below=.5ex] {$P$};
			\filldraw (3.5,0) circle (\r) node[below=.5ex] {$Q$};
			\begin{scope}[
					yshift=1.5ex,<->,
					shorten >=.3pt,shorten <=.3pt,
					fill=white,inner sep=1pt,
			]	
				\draw (0,0)   -- node[fill] {$D_1$} (1,0) ;
				\draw (1,0)   -- node[fill] {$D_2$} (3.5,0) ;
				\draw (3.5,0) -- node[fill] {$D_3$} (5,0) ;
			\end{scope}
		\end{tikzpicture}%
	}
	\caption{%
		Graphical representation of the relation between $\vect D$ and the pivot
		values $P$ and $Q$ on the unit interval.
	}
	\label{fig:relations-DPQ}
\end{figure}
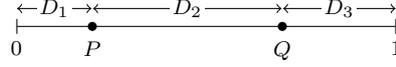

There is a close relation between $\vect I$, the number of small, medium and
large ordinary elements, and $\vect J$, the size of subproblems; 
we only have to add the sampled-out elements again before the recursive calls.
So we have $\vect J = \vect I + \vect t$ and 
$
		\Prob(\vect I = \vect i)  
	=
		\binom{i_1+t_1}{t_1} \binom{i_2+t_2}{t_2} \binom{i_3+t_3}{t_3}
		\big/
		\binom nk
$
by \wref{eq:prob-for-J-equals-j}.
Albeit valid, this form results in nasty sums with three binomials when we try
to compute expectations involving $\vect I$.

An alternative characterization of the distribution of $\vect I$ that is better
suited for our needs exploits that we have i.\,i.\,d.\ $\uniform(0,1)$
variables.
If we condition on the pivot \emph{values}, \ie, consider $P$ and $Q$
fixed, an ordinary element $U$ is small, if $U \in (0,P)$, 
medium if $U \in (P,Q)$ and 
large if $U \in (Q,1)$.
The lengths $\vect D = (D_1,D_2,D_3)$ of these three intervals 
(see \wref{fig:relations-DPQ}), 
thus are the \emph{probabilities} for an element to be small, medium or large,
respectively.
Note that this holds \emph{independently} of all other ordinary elements!
The partition sizes $\vect I$ are then obtained as the collective outcome of
$n-k$ independent drawings from this distribution, so 
conditional on $\vect D$, $\vect I$ is multinomially $\multinomial(n-k,\vect D)$
distributed.

With this alternative characterization, we have \emph{decoupled} the 
pivot \emph{ranks} (determined by $\vect I$) from the pivot
\emph{values}, which allows for a more elegant computation of expected values 
(see \wref{app:proof-of-lem-expectations}). 
This decoupling trick has (implicitly) been applied to the analysis of classic
Quicksort earlier, \eg, by \citet{neininger2001multivariate}.

\subsubsection{Distribution of Pivot Values}

The input array is initially filled with $n$ i.\,i.\,d.\ $\uniform(0,1)$
random variables 
from which we choose a sample $\{V_1,\ldots,V_k\} \subset \{U_1,\ldots,U_n\}$ 
of size $k$. 
The pivot values are then selected as order statistics of the sample:
$P \ce V_{(t_1+1)}$ and $Q \ce V_{(t_1+t_2+2)}$ 
(cf.\ \wref{sec:general-pivot-sampling}). 
In other words, $\vect D$ is the vector of \textsl{spacings}
induced by the order statistics $V_{(t_1+1)}$ and $V_{(t_1+t_2+2)}$ of $k$
i.\,i.\,d.\ $\uniform(0,1)$ variables $V_1,\ldots,V_k$, 
which is known to have a \weakemph{Dirichlet}
$\dirichlet(\vect t + 1)$ distribution
(\wref{pro:spacings-dirichlet-general-dimension} in the appendix).

\section{Expected Partitioning Costs}
\label{sec:expectations}

In \wref{sec:distributional-analysis}, we characterized the full distribution of
the costs of the first partitioning step. 
However, since those distributions are \emph{conditional} on other
random variables, we have to apply the \textsl{law of total expectation}.
By linearity of the expectation, it suffices to consider the
following summands:

\begin{lemma}
\label{lem:expectations}
	For pivot sampling parameter 
	$\vect t \in \N^3$ and
	partition sizes $\vect I \eqdist \multinomial(n-k,\vect D)$,
	based on random spacings $\vect D \eqdist \dirichlet(\vect t + 1)$, 
	the following (unconditional) expectations hold:
	\begin{align*}
			\E[I_j]
		&\wwrel=
			\frac{t_j+1}{k+1} (n-k) \,,
			\qquad\qquad 
			(j=1,2,3),
	\\
			\E\bigl[\bernoulli\bigl(\tfrac{I_3}{n-k}\bigr)\bigr]
		&\wwrel=
			\frac{t_3+1}{k+1}
			\wwrel{\wrel=} \Theta(1) \,,
			\qquad\quad
			(n\to\infty),
	\\
			\E\bigl[ \hypergeometric(I_3,I_1,n-k) \bigr]
		&\wwrel=
			\frac{(t_1+1)(t_3+1)}{(k+1)(k+2)} (n-k-1) \,,
	\\
			\E\bigl[ \hypergeometric(I_1+I_2, I_3, n-k) \bigr]
		&\wwrel=
			\frac{(t_1+t_2+2)(t_3+1)}{(k+1)(k+2)} (n-k-1) \;.
	\end{align*}
\end{lemma}

\noindent
Using known properties of the involved distributions, the proof is
elementary; see \wref{app:proof-of-lem-expectations} for details.

\section{Solution of the Recurrence}
\label{sec:solution-recurrence}

\begin{theorem}
\label{thm:leading-term-expectation-hennequin}
	Let $\E[C_n]$ be a sequence of numbers satisfying
	\wildtpageref[recurrence]{eq:ECn-recurrence}{\eqref} for a constant
	$\isthreshold \ge k$ and let the toll function $\E[T_n]$ be of the form $\E[T_n] = an+\Oh(1)$ for a
	constant $a$.
	Then we have
	\smash{$
			\E[C_n] 
		\sim 
			\frac{a}{\discreteEntropy} \, n \ln n
	$},
	where $\discreteEntropy$ is given by
	\wildtpageref[equation]{eq:discrete-entropy}{\eqref}.
\end{theorem}

\wref{thm:leading-term-expectation-hennequin} has first been proven by
\citet[Proposition~III.9]{hennequin1991analyse} using arguments on the
Cauchy-Euler differential equations that the recurrence
implies for the generating function of $\E[C_n]$. 
The tool box of handy and ready-to-apply theorems has grown considerably since
then. In \wref{app:CMT-solution}, we give a concise and elementary proof using
the \textsl{Continuous Master Theorem} \citep{Roura2001}:
we show that the distribution of the \emph{relative} subproblem sizes
converges to a \textsl{Beta distribution} and that then a continuous
version of the recursion tree argument allows to solve our recurrence.
An alternative tool closer to
\citeauthor{hennequin1991analyse}'s original arguments is offered by
\citet{Chern2002}.

\medskip\noindent
\wref{thm:expected-costs} now directly follows by using
\wref{lem:expectations} on the partitioning costs from
\wref{lem:distribution-partitioning-comparisons},
\ref{lem:distribution-partitioning-swaps}
and~\ref{lem:distribution-partitioning-bytecodes} and plugging the result into
\wref{thm:leading-term-expectation-hennequin}.

\section{Discussion\,---\,Asymmetries Everywhere}
\label{sec:asymmetries-everywhere}

\begin{figure}
	\newcommand\relativebarplot[2]{%
		\begin{tikzpicture}[
			xscale=0.85
		]
		\ifthenelse{ \lengthtest{#1pt<0pt} }{%
			\colorlet{barcolor}{green!80!black}%
			\def\nodepos{right}%
			\def\format##1##2|{{\bfseries\boldmath\color{green!60!black}$-$##2\%}}%
		}{
			\ifthenelse{ \lengthtest{#1pt>0pt} }{%
				\colorlet{barcolor}{red!80!black}%
				\def\nodepos{left}%
				\def\format##1##2|{\color{barcolor}$+$##1##2\%}%
			}{%
				\colorlet{barcolor}{white}%
				\def\nodepos{right}%
				\def\format##1##2|{}%
			}%
		}
		\useasboundingbox (-1,-.15) rectangle (1,.15) ;
		\ifthenelse{ \lengthtest{#1pt<#2pt} \AND \lengthtest{-#1pt<#2pt} }{
			\draw[thin,fill=barcolor] (0,.06) rectangle ($( #1 / #2 ,-.06)$) ;
			\node[\nodepos] at ($(0,0)$) {\tiny\format#1|}  ;
		}{
			\ifthenelse{ \lengthtest{#1pt>0pt} }{
				\draw[thin,fill=barcolor] 
					(0,-.06) -- ++(.8,0) coordinate (a) -- 
					++(.03,.12) coordinate (b) -- (0,.06) -- 
					++(0,-.12) -- cycle ;
				\draw[thin,fill=barcolor]
					(a) ++(.05,0) coordinate (c) --
					++(.03,.12) coordinate (d) -- (1.2,.06) --
					++(0,-.12) -- (c) -- cycle;
				\draw[thin,shorten >=-1pt,shorten <=-1pt] (a) -- (b) ; 
				\draw[thin,shorten >=-1pt,shorten <=-1pt] (c) -- (d) ;
				\node[\nodepos] at (0,0) {\tiny\format#1|} ;
			}{
				\draw[thin,fill=barcolor] 
					(0,-.06) -- ++(-.8,0) coordinate (a) -- 
					++(.03,.12) coordinate (b) -- (0,.06) -- 
					++(0,-.12) -- cycle ;
				\draw[thin,fill=barcolor]
					(a) ++(-.05,0) coordinate (c) --
					++(.03,.12) coordinate (d) -- (-1.2,.06) --
					++(0,-.12) -- (c) -- cycle;
				\draw[thin,shorten >=-1pt,shorten <=-1pt] (a) -- (b) ; 
				\draw[thin,shorten >=-1pt,shorten <=-1pt] (c) -- (d) ;
				\node[\nodepos] at (0,0) {\tiny\format#1|} ;
			}
		}
		\draw[very thick] (0,.15) -- (0,-.15) ; 
		\end{tikzpicture}%
	}%
	\newcommand\plotst[5]{%
		\node at (#1,#2) {%
			\scalebox{0.8}{\parbox{2cm}{%
			\centering%
			\relativebarplot{#3}{30}\\[-4pt]%
			\relativebarplot{#4}{10}\\[-4pt]%
			\relativebarplot{#5}{10}%
			}}%
		};
	}%
	\plaincenter{%

	\begin{tikzpicture}[
		scale=0.7,
		xscale=2.75,
		every node/.style={font=\scriptsize}
	]
	\fill[symmetriccolor] (1.5,1.5) rectangle ++(1,1) ;
	\plotst{0}{0}{68.7}{-17.7}{38.9}
	\plotst{0}{1}{32.5}{-8.84}{20.8}
	\plotst{0}{2}{18.8}{-1.36}{17.2}
	\plotst{0}{3}{15.0}{4.76}{20.4}
	\plotst{0}{4}{18.8}{9.52}{30.1}
	\plotst{0}{5}{32.5}{12.9}{49.7}
	\plotst{0}{6}{68.7}{15.0}{94.0}
	\plotst{1}{0}{32.5}{-11.6}{17.2}
	\plotst{1}{1}{11.4}{-4.76}{6.09}
	\plotst{1}{2}{3.86}{0.680}{4.57}
	\plotst{1}{3}{3.86}{4.76}{8.81}
	\plotst{1}{4}{11.4}{7.48}{19.7}
	\plotst{1}{5}{32.5}{8.84}{44.3}
	\plotst{2}{0}{18.8}{-8.16}{9.08}
	\plotst{2}{1}{3.86}{-3.40}{0.331}
	\plotst{2}{2}{0}{0}{0}
	\plotst{2}{3}{3.86}{2.04}{5.98}
	\plotst{2}{4}{18.8}{2.72}{22.0}
	\plotst{3}{0}{15.0}{-7.48}{6.37}
	\plotst{3}{1}{3.86}{-4.76}{-1.08}
	\plotst{3}{2}{3.86}{-3.40}{0.331}
	\plotst{3}{3}{15.0}{-3.40}{11.1}
	\plotst{4}{0}{18.8}{-9.52}{7.47}
	\plotst{4}{1}{11.4}{-8.84}{1.54}
	\plotst{4}{2}{18.8}{-9.52}{7.47}
	\plotst{5}{0}{32.5}{-14.3}{13.6}
	\plotst{5}{1}{32.5}{-15.6}{11.8}
	\plotst{6}{0}{68.7}{-21.8}{32.0}
	\foreach \t in {0,...,6}{
		\node at (\t,-.9) {$t_1=\t$} ;
		\node at (-.75,\t) {$t_2=\t$} ;
	}
	\begin{scope}[overlay]
		\draw (-1,-0.5) -- ++(7.5,0) ;
		\draw (-1,0.5) -- ++(7.5,0) -- ++(0,-1.7) ;
		\draw (-1,1.5) -- ++(6.5,0) -- ++(0,-2.7) ;
		\draw (-1,2.5) -- ++(5.5,0) -- ++(0,-3.7) ;
		\draw (-1,3.5) -- ++(4.5,0) -- ++(0,-4.7) ;
		\draw (-1,4.5) -- ++(3.5,0) -- ++(0,-5.7) ;
		\draw (-1,5.5) -- ++(2.5,0) -- ++(0,-6.7) ;
		\draw (-1,6.5) -- ++(1.5,0) -- ++(0,-7.7) ;
		\draw            (-0.5,6.5) -- ++(0,-7.7) ;
	\end{scope}
	\draw[ultra thick] (2.5,0.5) rectangle ++(1,1) ;
	\begin{scope}[overlay]
			\node at (5.5,5) {%
			\scalebox{0.9}{%
			\parbox{1.3cm}{\raggedleft%
				$1/\discreteEntropy$:\\
				$a_C$:\\
				$a_C / \discreteEntropy$:
			}}\hspace{-4pt}%
			\scalebox{1.2}{%
			\parbox{2cm}{%
			\centering%
				\relativebarplot{15.0}{30}\\[-4pt]%
				\relativebarplot{-7.48}{10}\\[-4pt]%
				\relativebarplot{6.37}{10}%
			}}%
		};
	\end{scope}
	\end{tikzpicture}%
	}\\[-1.5\baselineskip]%
	\caption{%
		Inverse of discrete entropy (top), number of comparisons per partitioning step
		(middle) and overall comparisons (bottom)
		for all $\vect t$ with $k=8$, relative to the tertiles case $\vect t =
		(2,2,2)$. 
	}
	\label{fig:relative-cmps-8}
\end{figure}
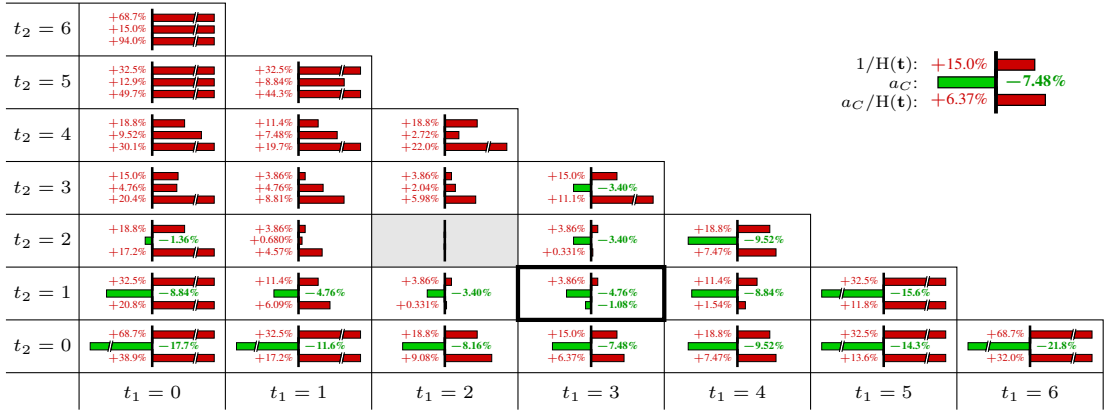

\begin{table}
	\setlength\subfigbottomskip{-1ex}
	\plaincenter{%
	\footnotesize%
	\setlength\tabcolsep{0.25em}%
	\subtable[$a_C / \discreteEntropy$]{%
		\begin{tabular}{c|cccc}
			${}_{t_{1}\!\!\!}\diagdown{}^{\!\! t_{2}}$
				& 0 & 1 & 2 & 3 \\
			\hline 
			0 & 1.9956 & 1.8681 & 2.0055 & 2.4864 \\
			1 & 1.7582 & \textbf{1.7043} \cellcolor{symmetriccolor}
			                                & 1.9231 \\
			2 & 1.7308 & 1.7582 &  \\
			3 & 1.8975 & \\
		\end{tabular}%
	}%
	\hfill%
	\subtable[$a_S / \discreteEntropy$]{%
	\begin{tabular}{c|cccc}
			${}_{t_{1}\!\!\!}\diagdown{}^{\!\! t_{2}}$
				& 0 & 1 & 2 & 3 \\
			\hline 
			0 & 0.4907 & 0.4396 & 0.4121 & \textbf{0.3926} \\
			1 & 0.6319 & 0.5514 \cellcolor{symmetriccolor} & 0.5220 \\
			2 & 0.7967 & 0.7143 \\
			3 & 1.0796 \\
		\end{tabular}%
	}%
	\hfill%
	\subtable[$a_{\bytecodes} / \discreteEntropy$]{%
		\begin{tabular}{c|cccc}
			${}_{t_{1}\!\!\!}\diagdown{}^{\!\! t_{2}}$
				& 0 & 1 & 2 & 3 \\
			\hline 
			0 & 20.840 & \textbf{18.791} & 19.478 & 23.293 \\
			1 & 20.440 & 19.298 \cellcolor{symmetriccolor}& 21.264 \\
			2 & 22.830 & 22.967 \\
			3 & 29.378 \\
		\end{tabular}%
	}%
	}
	\caption{%
		$\frac{a_C}{\discreteEntropy}$, 
		$\frac{a_S}{\discreteEntropy}$ and
		$\frac{a_{\bytecodes}}{\discreteEntropy}$ 
		for all $\vect t$ with $k=5$.
		Rows resp.\ columns give $t_1$ and $t_2$; $t_3$ is then $k-2-t_1-t_2$.
		The symmetric choice $\vect t = (1,1,1)$ is shaded, the minimum is printed in
		bold.%
	}
	\label{tab:results-k5}
\end{table}

With \wref{thm:expected-costs}, we can find the optimal sampling
parameter $\vect t$ for any given sample size~$k$.
As an example, \wref{fig:relative-cmps-8} shows how $\discreteEntropy$, $a_C$
and the overall number of comparisons behave for all possible $\vect t$ with
sample size $k=8$:
the discrete entropy decreases symmetrically as we move away
from the center $\vect t = (2,2,2)$; this corresponds to the effect of less
evenly distributed subproblem sizes. 
The individual partitioning steps, however, are cheap for \emph{small} values of
$t_2$ and optimal in the extreme point $\vect t = (6,0,0)$.
For minimizing the \emph{overall} number of comparisons\,---\,the ratio of
latter two numbers\,---\,we have to find a suitable trade-off between the
center and the extreme point $(6,0,0)$; in this case the minimal total number of
comparisons is achieved with $\vect t = (3,1,2)$.

Apart from this trade-off between the evenness of subproblem sizes and the
number of comparisons per partitioning, 
\wref{tab:results-k5} shows that the optimal choices for $\vect t$ w.\,r.\,t.\
comparisons, swaps and Bytecodes heavily differ.
The partitioning costs are, in fact, in \emph{extreme conflict} with each other:
for all $k\ge 2$, the minimal values of $a_C$, $a_S$ and $a_{\bytecodes}$ among
all choices of $\vect t $ for sample size $k$ are attained for
$\vect t = (k-2,0,0)$, 
$\vect t = (0,k-2,0)$ and 
$\vect t = (0,0,k-2)$, respectively.
Intuitively this is so, as the strategy minimizing partitioning costs in
isolation is to make the cheapest path through the partitioning loop execute as
often as possible, which naturally leads to extreme choices for $\vect t$.
It then depends on the actual numbers, where the total
costs are minimized.
It is thus not possible to minimize all cost measures at
once, and the rivaling effects described above make it hard to reason
about optimal parameters merely on a qualitative level.
The number of executed Bytecode instructions is certainly more closely related
to actual running time than the pure number of comparisons and swaps, 
while it remains platform independent and deterministic.%
\footnote{%
	Counting the number of executed Bytecode instructions still ignores many  
	important effects on actual running time, \eg, costs of branch
	mispredictions in pipelined execution, cache misses and the influence of
	just-in-time compilation.
} 
We hope that the sensitivity of the optimal sampling parameter to the chosen
cost measure renews the interest in instruction-level analysis in the style
of Knuth.
Focusing only on abstract cost measures leads to
\emph{suboptimal} choices in Yaroslavskiy's Quicksort!

It is interesting to note in this context that the implementation in
Oracle's Java 7 runtime library\,---\,which uses $\vect t =
(1,1,1)$\,---\,executes asymptotically \emph{more} Bytecodes
(on random permutations) than \generalYarostM with $\vect t=(0,1,2)$, 
despite using the same sample size $k=5$.
Whether this also results in a performance gain in practice, however, depends
on details of the runtime environment \citep{Wild2013Alenex}.

\paragraph{Continuous ranks}
It is natural to ask for the optimal \emph{relative ranks} of $P$ and $Q$
if we are not constrained by the discrete nature of pivot sampling.
In fact, one might want to choose the sample size depending on those optimal
relative ranks to find a discrete order statistic that falls close to the
continuous optimum.

We can compute the optimal relative ranks by considering the limiting
behavior of $\generalYarostM$ as $k\to\infty$. 
Formally, we consider the following family of algorithms:
let \smash{$(\ui{t_l}j)_{j\in\N}$} for $l=1,2,3$ be three
sequences of non-negative integers and set $\ui kj \ce \ui{t_1}j + \ui{t_2}j +
\ui{t_3}j + 2$ for every $j\in\N$. 
Assume that we have $\ui kj \to \infty$ and
${\ui{t_l}j}/{\ui kj} \to \tau_l$ with $\tau_l\in[0,1]$ for $l=1,2,3$ as
$j\to\infty$.
Note that we have $\tau_1 + \tau_2 + \tau_3 = 1$ by definition.
For each $j\in\N$, we can apply \wref{thm:expected-costs} 
for $\generalYaros{\ui{\vect t}j}{\isthreshold}$ and then
consider the limiting behavior of the total costs for $j\to\infty$.%
\footnote{%
	Letting the sample size go to infinity implies non-constant overhead per
	partitioning step for our implementation, which is not negligible
	any more. 
	For the analysis here, you can assume an oracle that provides us
	with the desired order statistic in $\Oh(1)$. 
}
For \discreteEntropy, \wildref[equation]{eq:limit-g-entropy}{\eqref} shows
convergence to the entropy function
$\contentropy = -\sum_{l=1}^3 \tau_l \ln(\tau_l)$ and 
for the numerators $a_C$, $a_S$ and $a_\bytecodes$, it is easily seen that
\begin{align*}
		\ui{a_C}j 
	&\wwrel\to 
		\like[l]{a^*_\bytecodes}{a^*_C} 
		\wrel\ce 
		1 + \tau_2 + (2\tau_1 + \tau_2) \tau_3 \,, 
\\		\ui{a_S}j
	&\wwrel\to
		\like[l]{a^*_\bytecodes}{a^*_S} 
		\wrel\ce 
		\tau_1 + (\tau_1 + \tau_2)\tau_3 \,,
\\		\ui{a_\bytecodes}j
	&\wwrel\to
		a^*_\bytecodes 
		\wrel\ce	
		10 + 13\tau_1 + 5\tau_2 + (\tau_1+\tau_2)(\tau_1+11\tau_3)
	\;.
\end{align*}
Together, the overall number of comparisons, swaps and Bytecodes converge to
$a^*_C / \contentropy$, $a^*_S / \contentropy$ resp.\
$a^*_\bytecodes / \contentropy$;
see \wref{fig:3dplot-limit-total-costs} for plots.
\begin{figure}
	\setlength\subfigbottomskip{0pt}
	\newcommand\contourplot[5]{%
		\resizebox{.32\linewidth}!{
			\begin{tikzpicture}
				\begin{axis}[
						width=.45\linewidth,
						height=.45\linewidth,
						font=\footnotesize,
						enlargelimits=0.06,
						xmin=0,xmax=1,ymin=0,ymax=1,
						xtick={0,0.2,...,1},
						ytick={0,0.2,...,1},
						y label style={rotate=-90},
					]
					\addplot graphics [xmin=0,xmax=1,ymin=0,ymax=1] 
						{pslt-pics/#1-inf-entropy} ;
					\draw[thin] (axis cs:0,0) -- (axis cs:1,0) -- (axis cs:0,1) -- cycle ;
					\ifthenelse{\equal{#2}{}}{}{%
 						\draw[semithick,black,<-,shorten <=1.5pt] 
 							(axis cs:#2) -- +(45:{6em+#5}) 
 							node[anchor=west,inner sep=1pt] {$#3$} ;
 					}
 					\draw[semithick,black,<-,shorten <=1.5pt] 
 						(axis cs:0.3333,0.3333) -- +(45:6em-#5) 
 						node[anchor=west,inner sep=1pt]	{$#4$} ;
				\end{axis}
			\end{tikzpicture}
		}%
	}
	\subfigure[$a^*_C / \contentropy$]{%
		\contourplot{comparisons}{0.4288,0.2688}{1.4931}{1.5171}{-0.5	em}%
	}\hfill%
	\subfigure[$a^*_S / \contentropy$]{%
		\contourplot{swaps}{}{}{0.5057}{0pt}%
	}\hfill%
	\subfigure[$a^*_{\bytecodes} / \contentropy$]{%
		\contourplot{bytecodes}{0.2068,0.3486}{16.383}{16.991}{1em}%
	}%
	\caption{%
		Contour plots for the limits of the leading term coefficient of the
		overall number of comparisons, swaps and executed Bytecode instructions,
		as functions of $\vect\tau$.
		$\tau_1$ and $\tau_2$ are given on $x$- and $y$-axis, respectively, which
		determine $\tau_3$ as $1-\tau_1-\tau_2$. 
		Black dots mark global minima, white dots show the center point
		$\tau_1=\tau_2=\tau_3=\frac13$. 
		(For swaps no minimum is attained in the open simplex, see main text).
		Black dashed lines are level lines connecting ``equi-cost-ant'' points, \ie\
		points of equal costs. 
		White dotted lines mark points of equal entropy $\contentropy$. 
	}
	\label{fig:3dplot-limit-total-costs}
\end{figure}
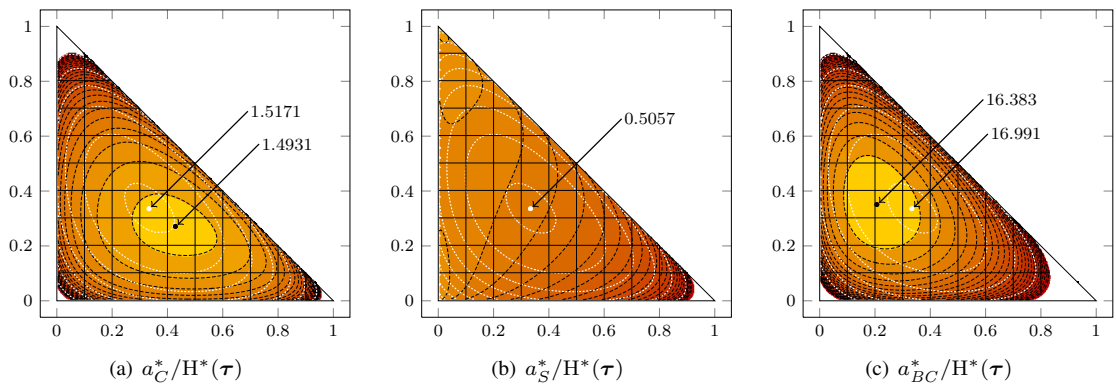
We could not find a way to compute the minima of these functions analytically.
However, all three functions have isolated minima that can be
approximated well by numerical methods.

The number of comparisons is minimized for 
$
		\vect\tau^*_C 
	\approx
		(0.428846,0.268774,0.302380)
$.
For this choice, the expected number of comparisons is
asymptotically $1.4931 \, n\ln n$.
For swaps, the minimum is not attained inside the open simplex, but for
the extreme points $\vect\tau^*_S = (0,0,1)$ and 
$\vect\tau_S^{*\prime} = (0,1,0)$.
The minimal value of the coefficient is $0$, so the expected number of swaps
drops to $o(n\ln n)$ for these extreme points.
Of course, this is a very bad choice w.\,r.\,t.\ other cost measures, \eg,
the number of comparisons becomes quadratic, which
again shows the limitations of tuning an algorithm to one of its
basic operations in isolation.
The minimal asymptotic number of executed Bytecodes of
roughly $16.3833 \, n\ln n$ is obtained for 
$
		\vect\tau^*_\bytecodes 
	\approx
		(0.206772,0.348562,0.444666)
$.

We note again that the optimal choices
heavily differ depending on the employed cost measure and that the minima differ
significantly from the symmetric choice $\vect\tau=(\frac13,\frac13,\frac13)$.

\section{Conclusion}
\label{sec:conclusion}

In this paper, we gave the precise leading term asymptotic of the average costs
of Quicksort with Yaroslavskiy's dual-pivot partitioning method and selection of
pivots as arbitrary order statistics of a constant size sample.
Our results confirm earlier empirical findings
\citep{javacoredevel2010,Wild2013Alenex} that the inherent asymmetries of the
partitioning algorithm call for a systematic skew in selecting the
pivots\,---\,the tuning of which requires a quantitative understanding of the
delicate trade-off between partitioning costs and the distribution of
subproblem sizes for recursive calls.
Moreover, we have demonstrated that this tuning process is very sensitive to the
choice of suitable cost measures, which firmly suggests a detailed
analyses in the style of Knuth, 
instead of focusing on the number of comparisons and swaps only.

\paragraph{Future work}
A natural extension of this work would be the computation of
the linear term of costs, which is not negligible for moderate $n$.
This will require a much more detailed analysis as sorting the samples and
dealing with short subarrays contribute to the linear term of costs, but then
allows to compute the optimal choice for \isthreshold, as well.
While in this paper only expected values were considered,
the distributional analysis of \wref{sec:distributional-analysis} can be used as a
starting point for analyzing the distribution of overall costs.
Yaroslavskiy's partitioning can also be used in Quickselect
\citep{WildMahmoud2013arxiv}; the effects of generalized pivot sampling there
are yet to be studied.
Finally, other cost measures, like the number of
symbol comparisons \citep{vallee2009symbolComparisons,Fill2012},
would be interesting to analyze.

\begin{small}
\bibliography{quicksort-refs}
\end{small}

\clearpage
\appendix
\section*{Appendix}
\section{Index of Used Notation}
\label{app:notations}
\def\mydots{\xleaders\hbox to.75em{\hfill.\hfill}\hfill}

\newlength\tmpLenNotations
\newenvironment{notations}[1][10em]{%
	\small
	\newcommand\notationentry[1]{%
		\settowidth\tmpLenNotations{##1}%
		\ifthenelse{\lengthtest{\tmpLenNotations > \labelwidth}}{%
			\parbox[b]{\labelwidth}{%
				\makebox[0pt][l]{##1}\\%
			}%
		}{%
			\mbox{##1}%
		}%
		\mydots\relax%
	}%
	\begin{list}{}{%
		\setlength\labelsep{0em}%
		\setlength\labelwidth{#1}%
		\setlength\leftmargin{\labelwidth+\labelsep+1em}%
		\renewcommand\makelabel{\notationentry}%
	}
	\newcommand\notation[1]{\item[##1]}
	\raggedright
}{%
	\end{list}
}

In this section, we collect the notations used in this paper.
(Some might be seen as ``standard'', but we think
including them here hurts less than a potential misunderstanding caused by
omitting them.)

\subsection*{Generic Mathematical Notation}
\begin{notations}
\notation{$\ln n$}
	natural logarithm.
\notation{$\vect x$}
	to emphasize that $\vect x$ is a vector, it is written in \textbf{bold};\\
	components of the vector are not written in bold: $\vect x = (x_1,\ldots,x_d)$.
\notation{$X$}
	to emphasize that $X$ is a random variable it is Capitalized.
\notation{$\harm{n}$}
	$n$th harmonic number; $\harm n = \sum_{i=1}^n 1/i$.
\notation{$\dirichlet(\vect \alpha)$}
	Dirichlet distributed random variable, 
	$\vect \alpha \in \R_{>0}^d$.
\notation{$\multinomial(n,\vect p)$}
	multinomially distributed random variable; 
	$n\in\N$ and $\vect p \in [0,1]^d$ with $\sum_{i=1}^d p_i = 1$.
\notation{$\hypergeometric(k,r,n)$}
	hypergeometrically distributed random variable;
	$n\in\N$, $k,r,\in\{1,\ldots,n\}$.  
\notation{$\bernoulli(p)$}
	Bernoulli distributed random variable;
	$p\in[0,1]$.
\notation{$\uniform(a,b)$}
	uniformly in $(a,b)\subset\R$ distributed random variable. 
\notation{$\BetaFun(\alpha_1,\ldots,\alpha_d)$}
	$d$-dimensional Beta function; defined in
	\wildpageref[equation]{eq:def-beta-function}{\eqref}.
\notation{{$\E[X]$}}
	expected value of $X$; we write $\E[X\given Y]$ for the conditional expectation
	of $X$ given $Y$.
\notation{$\Prob(E)$, $\Prob(X=x)$}
	probability of an event $E$ resp.\ probability for random variable $X$ to
	attain value $x$.
\notation{$X\eqdist Y$}
	equality in distribution; $X$ and $Y$ have the same distribution.
\notation{$X_{(i)}$}
	$i$th order statistic of a set of random variables $X_1,\ldots,X_n$,\\
	\ie, the $i$th smallest element of $X_1,\ldots,X_n$. 
\notation{$\indicator{E}$}
	indicator variable for event $E$, \ie, $\indicator{E}$ is $1$ if $E$
	occurs and $0$ otherwise.
\notation{$a^{\underline b}$, $a^{\overline b}$}
	factorial powers notation of \citep{ConcreteMathematics}; 
	``$a$ to the $b$ falling resp.\ rising''.
\end{notations}

\subsection*{Input to the Algorithm}
\begin{notations}
\notation{$n$}
	length of the input array, \ie, the input size.
\notation{\arrayA}
	input array containing the items $\arrayA[1],\ldots,\arrayA[n]$ to be
	sorted; initially, $\arrayA[i] = U_i$.
\notation{$U_i$}
	$i$th element of the input, \ie, initially $\arrayA[i] = U_i$.\\
	We assume $U_1,\ldots,U_n$ are i.\,i.\,d.\ $\uniform(0,1)$ distributed.
\end{notations}

\subsection*{Notation Specific to the Algorithm}
\begin{notations}
\notation{$\vect t \in \N^3$}
	pivot sampling parameter, see \wpref{sec:general-pivot-sampling}.
\notation{$k=k(\vect t)$}
	sample size; defined in terms of $\vect t$ as $k(\vect t) = t_1+t_2+t_3+2$.
\notation{\isthreshold}
	Insertionsort threshold; for $n\le\isthreshold$, Quicksort recursion is
	truncated and we sort the subarray by Insertionsort.
\notation{$\generalYarostM$}
	abbreviation for dual-pivot Quicksort with Yaroslavskiy's partitioning method,
	where pivots are chosen by generalized pivot sampling with parameter $\vect t$
	and where we switch to Insertionsort for subproblems of size at most
	\isthreshold.
\notation{$\iscost_n$}
	(random) costs of sorting a random permutation of size $n$ with Insertionsort. 
\notation{$\vect V \in \N^k$}
	(random) sample for choosing pivots in the first partitioning step.
\notation{$P$, $Q$}
	(random) values of chosen pivots  in the first partitioning step.
\notation{small element}
	element $U$ is small if $U<P$.
\notation{medium element}
	element $U$ is medium if $P<U<Q$.
\notation{large element}
	element $U$ is large if $Q < U$.
\notation{sampled-out element}
	the $k-2$ elements of the sample that are \emph{not} chosen as pivots.	
\notation{ordinary element}
	the $n-k$ elements that have not been part of the sample.
\notation{partitioning element}
	all ordinary elements and the two pivots.
\notation{$k$, $g$, $\ell$}
	index variables used in Yaroslavskiy's partitioning method, see
	\wpref{alg:partition}. 
\notation{$\positionsets{K}$, $\positionsets{G}$}
	set of all (index) values attained by pointers $k$ resp.\ $g$ during the first
	partitioning step; see \wpref{sec:yaroslavskiys-partitioning-method} and proof of
	\wpref{lem:distribution-partitioning-comparisons}.
\notation{$\numberat{c}{P}$}
	$c\in\{s,m,l\}$, $\positionsets{P} \subset \{1,\ldots,n\}$\\
	(random) number of $c$-type ($s$mall, $m$edium or $l$arge) elements 
	that are initially located at positions in $\positionsets{P}$, \ie,
	$
		\numberat{c}{P} \wrel= 
		\bigl|\{
			i \in \positionsets{P} : U_i \text{ has type } c
		\}\bigr|. 
	$
\notation{$\latK$, $\satK$, $\satG$}
	see $\numberat{c}{P}$
\notation{$\chi$}
	(random) point where $k$ and $g$ first meet.
\notation{$\delta$}
	indicator variable of the random event that $\chi$ is on a large
	element, \ie, $\delta = \indicator{U_\chi > Q}$.
\notation{$C_n$, $S_n$, $\bytecodes_n$}
	(random) number of comparisons\,/\,swaps\,/\,Bytecodes of \generalYarostM on a
	random permutation of size $n$;
	in \wref{sec:recurrence-quicksort}, $C_n$ is used as general placeholder
	for any of the above cost measures.  
\notation{\mbox{$\toll{C}$, $\toll{S}$, $\toll{\bytecodes}$}}
	(random) number of comparisons\,/\,swaps\,/\,Bytecodes of the first
	partitioning step of \generalYarostM on a random permutation of size $n$;\\
	$\toll[n]{C}$, $\toll[n]{S}$ and $\toll[n]{\bytecodes}$ when we want to
	emphasize dependence on $n$.
\notation{$a_C$, $a_S$, $a_{\bytecodes}$}
	coefficient of the linear term of $\E[\toll[n]{C}]$, $\E[\toll[n]{S}]$ and
	$\E[\toll[n]{\bytecodes}]$; see \wpref{thm:expected-costs}.
\notation{$\discreteEntropy$}
	discrete entropy; defined in
	\wildpageref[equation]{eq:discrete-entropy}{\eqref}.
\notation{{$\contentropy[\vect p]$}}
	continuous (Shannon) entropy with basis $e$; defined in
	\wildpageref[equation]{eq:limit-g-entropy}{\eqref}.
\notation{$\vect J\in\N^3$}
	(random) vector of subproblem sizes for recursive calls;\\
	for initial size $n$, we have $\vect J \in \{0,\ldots,n-2\}^3$ with
	$J_1+J_2+J_3 = n-2$.
\notation{$\vect I\in\N^3$}
	(random) vector of partition sizes, \ie, the number of small, medium resp.\
	large \emph{ordinary} elements;
	for initial size $n$, we have $\vect I \in \{0,\ldots,n-k\}^3$ with
	$I_1+I_2+I_3 = n-k$;\\
	$\vect J = \vect I + \vect t$ and conditional on $\vect D$ we
	have $\vect I \eqdist \multinomial(n-k,\vect D)$.
\notation{{$\vect D\in[0,1]^3$}}
	(random) spacings of the unit interval $(0,1)$ induced by the pivots $P$ and
	$Q$, \ie, $\vect D = (P,Q-P,1-Q)$;
	$\vect D \eqdist \dirichlet(\vect t + 1)$.
\notation{$a^*_C$, $a^*_S$, $a^*_{\bytecodes}$}
	limit of $a_C$, $a_S$, resp.\ $a_{\bytecodes}$ for the optimal sampling
	parameter $\vect t$ when $k\to\infty$.
\notation{$\vect\tau_C^*$, $\vect\tau_S^*$, $\vect\tau_{\bytecodes}^*$}
	optimal limiting ratio $\vect t / k \to \vect \tau_C^*$ such that $a_C \to
	a^*_C$ (resp.\ for $S$ and $\bytecodes$).
\end{notations}

\clearpage
\section{Detailed Pseudocode}
\label{app:algorithms}

\subsection{Implementing Generalized Pivot Sampling}
\label{sec:generalized-pivot-sampling-implementation}

While extensive literature on the analysis of (single-pivot) Quicksort with
pivot sampling is available, most works do not specify the pivot
selection process in detail.%
\footnote{%
	Noteworthy exceptions are \citeauthor{Sedgewick1977}'s seminal works
	which give detailed code for the median-of-three strategy
	\citep{Sedgewick1975,Sedgewick1978} and \citeauthor{Bentley1993}'s
	influential paper on engineering a practical sorting
	method~\citep{Bentley1993}. 
	\Citeauthor{Martinez2001} describe a general approach of which they state
	that randomness is \emph{not} preserved, but in their analysis, they “disregard
	the small amount of sortedness [\,\dots] yielding at least a good approximation”
	\citep[Section~7.2]{Martinez2001}.
}
The usual justification is that, in any case, we only draw pivots a
\emph{linear} number of times and from a constant size sample. So for the
leading term asymptotic, the costs of pivot selection are negligible, 
and hence also the precise way of how selection is done is not important.

There is one caveat in the argumentation: 
Analyses of Quicksort usually rely on setting up a recurrence equation of
expected costs that is then solved (precisely or asymptotically).
This in turn requires the algorithm to \emph{preserve} the distribution of
input permutations for the subproblems subjected to recursive
calls\,---\,otherwise the recurrence does not hold.
Most partitioning algorithms, including the one of Yaroslavskiy, have the
desirable property to preserve randomness \citep{Wild2012}; but this is not
sufficient! 
We also have to make sure that the main procedure of Quicksort does not 
alter the distribution of inputs for recursive calls;
in connection with elaborate pivot sampling algorithms, this is harder to
achieve than it might seem at first sight.

For these reasons, the authors felt the urge to include a minute discussion
of how to implement the generalized pivot sampling scheme of
\wref{sec:general-pivot-sampling} in such a way that the recurrence equation
remains \emph{precise}.%
\footnote{%
	Note that the resulting implementation has to be considered ``academic'':
	While it is well-suited for precise analysis, it will look somewhat peculiar
	from a practical point of view and productive use is probably not
	to be recommended.
}
We have to address the following questions:

\oldparagraph{Which elements do we choose for the sample?}

In theory, a \emph{random} sample produces the most reliable results and
also protects against worst case inputs.
The use of a random pivot for classic Quicksort has been considered right from
its invention \citep{Hoare1961} and is suggested as a general strategy to deal
with biased data \citep{Sedgewick1978}.

However, all programming libraries known to the authors actually avoid the
additional effort of drawing random samples. 
They use a set of deterministically selected positions of the array, instead;
chosen to give reasonable results for common special cases like almost sorted arrays.
For example, the positions used in Oracle's Java~7 implementation are depicted
in \wref{fig:sample-choice-jre7}.

For our analysis, the input consists of i.\,i.\,d.\ random variables, so
\emph{all} subsets (of a certain size) have the same distribution.
We might hence select the positions of sample elements such that they
are convenient for our (analysis) purposes.
For reasons elaborated in \wref{sec:randomness-preservation} below, we have to 
\emph{exclude} sampled-out elements from partitioning to keep analysis feasible,
and therefore, our implementation uses the $t_1+t_2+1$ leftmost
and the $t_3+1$ rightmost elements of the array as sample, as
illustrated in \wref{fig:sample-choice-generalized-yaroslavskiy}.
Then, partitioning can be simply restricted to the range between the two parts
of the sample, namely positions $t_1+t_2+2$ through $n-t_3-1$.

\begin{figure}
	\plaincenter{%
	\begin{tikzpicture}
	[scale=0.5,baseline=2,every node/.style={font={\footnotesize}}]
		\draw (0,0) rectangle (3,1) ; 
		\draw[decoration=brace,decorate] 
		  		(0,1.25) -- node[above] {\scriptsize$\frac3{14}n$} (3,1.25) ; 
		\foreach \x in {0,3,6,9} {
		  \draw[fill=black!10,very thick] (3+\x,0) rectangle (4+\x,1) ; 
		  \draw (4+\x,0) rectangle (6+\x,1) ;
		  \draw[decoration=brace,decorate] 
		  		(4+\x,1.25) -- node[above]{\scriptsize$\frac17n$} (6+\x,1.25) ; 
		}
		\draw[fill=black!10,very thick] (15,0) rectangle (16,1) ;
	    \draw (16,0) rectangle (19,1) ;
	    \draw[decoration=brace,decorate] 
		  		(16,1.25) -- node[above]{\scriptsize$\frac3{14}n$} (19,1.25) ;
	    \node at (6.5,-0.5) {$P$};
	    \node at (12.5,-0.5) {$Q$};
		\foreach \x in {1,...,5} {
			\node at (\x*3+0.5,0.5) {$V_{\x}$} ;
		}
	\end{tikzpicture}
	}
	\caption{%
		The five sample elements in Oracle's Java~7 implementation of
		Yaroslavskiy's dual-pivot Quicksort are chosen such that
		their distances are approximately as given above. 
	}
	\label{fig:sample-choice-jre7}
\end{figure}
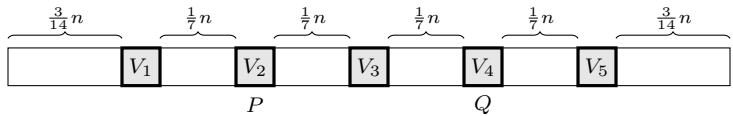

\begin{figure}[tbhp]
	\plaincenter{%
	\begin{tikzpicture}[scale=0.5,every node/.style={font=\footnotesize}]
		\begin{scope}[very thick,fill=black!10]
			\filldraw (0,0) rectangle (6,1);
			\filldraw (15,0) rectangle (20,1); 
		\end{scope}
		\begin{scope}[thick]
			\draw (3,0) -- (3,1);
			\draw (4,0) -- (4,1);
			\draw (16,0) -- (16,1); 
		\end{scope}
	
		\draw[thin] (0,0) grid (20,1) ;
		\foreach \x in {1,...,20}  \node at (\x-0.5,1.5) {\scriptsize\x} ;
	
		\begin{scope}[thick,decoration=brace, yshift=-5pt]
			\draw[decorate] (3,0) -- node[below]{$t_1$} (0,0) ;
			\draw[decorate] (6,0) -- node[below]{$t_2$} (4,0) ;
			\draw[decorate] (20,0) -- node[below]{$t_3$} (16,0) ;
		\end{scope}
		\node at (3.5,-0.5) {$P$};
		\node at (15.5,-0.5) {$Q$};
		
		\foreach \x in {1,...,6} {
			\node at (\x-0.5,0.5) {$V_{\x}$} ; 
		}
		\foreach \x in {7,...,11} {
			\node at (16-7+\x-0.5,0.5) {$V_{\x}$} ; 
		}
		
	\end{tikzpicture}
	}
	\caption{%
		Location of the sample in our implementation of \generalYarostM
		with $\vect t = (3,2,4)$.
		Only the non-shaded region $\protect\arrayA[7..15]$ is subject to partitioning. 
	}
	\label{fig:sample-choice-generalized-yaroslavskiy}
\end{figure}
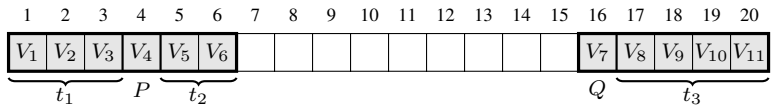

\oldparagraph{How do we select the desired order statistics from the sample?}
Finding a given order statistic of a list of elements is known as the
\weakemph{selection problem} and can be solved by specialized algorithms like
Quickselect.
Even though these selection algorithms are superior by far on large lists,
selecting pivots from a reasonably small sample is most efficiently done by
fully sorting the whole sample with an elementary sorting method.
Once the sample has been sorted, we find the pivots in $\arrayA[t_1+1]$ and
$\arrayA[n-t_3]$, respectively.

We will use an Insertionsort variant for sorting samples.
Note that the implementation has to “jump” across the gap
between the left part and the right part of the sample.
\wref{alg:samplesort-left} and its symmetric cousin \wref{alg:samplesort-right}
do that by ignoring the gap for all index variables and then correct for the gap
whenever the array is actually accessed.

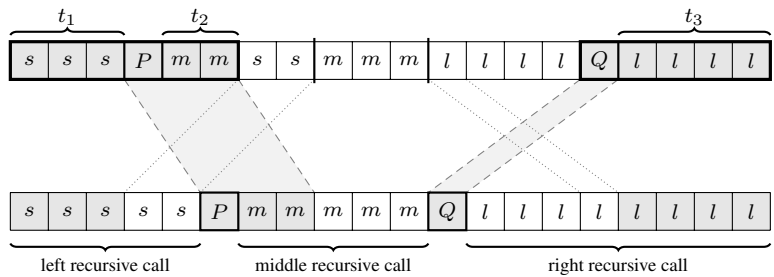
\begin{figure}
	\plaincenter{%
	\begin{tikzpicture}[scale=0.5,every node/.style={font=\footnotesize}]
		\begin{scope}[very thick,fill=black!10]
			\filldraw (0,0) rectangle (6,1);
			\filldraw (15,0) rectangle (20,1); 
		\end{scope}
		\begin{scope}[thick]
			\draw (3,0) -- (3,1);
			\draw (4,0) -- (4,1);
			\draw (16,0) -- (16,1); 
		\end{scope}
	
		\draw[thin] (0,0) grid (20,1) ;
		
		\begin{scope}[thick,decoration=brace, yshift=5pt]
			\draw[decorate] (0,1) -- node[above]{$t_1$} (3,1) ;
			\draw[decorate] (4,1) -- node[above]{$t_2$} (6,1) ;
			\draw[decorate] (16,1) -- node[above]{$t_3$} (20,1) ;
		\end{scope}
		\node at (3.5,0.5) {$P$};
		\node at (15.5,0.5) {$Q$};
		
		\foreach \x in {1,...,3,7,8} 			\node at (\x-0.5,0.5) {$s$} ; 
		\foreach \x in {5,6, 9,10,11} 			\node at (\x-0.5,0.5) {$m$} ;
		\foreach \x in {12,...,15,17,18,19,20}	\node at (\x-0.5,0.5) {$l$} ;
		
		\draw[thick] ( 8,-0.1) -- ( 8,1.1) ;
		\draw[thick] (11,-0.1) -- (11,1.1) ;
		
		\begin{scope}[yshift=-4cm]
			\begin{scope}[fill=black!10]
				\fill (0,0) rectangle (3,1);
				\fill (5,0) rectangle (8,1);
				\fill (11,0) rectangle ++(1,1);
				\fill (16,0) rectangle (20,1); 
			\end{scope}
			\begin{scope}[thick]
				\draw (5,0) rectangle ++(1,1);
				\draw (11,0) rectangle ++(1,1);
			\end{scope}
		
			\draw[thin] (0,0) grid (20,1) ;
						
			\node at (6-0.5,0.5) {$P$};
			\node at (12-0.5,0.5) {$Q$};
			
			\foreach \x in {1,...,5} 	\node at (\x-0.5,0.5) {$s$} ; 
			\foreach \x in {7,...,11}	\node at (\x-0.5,0.5) {$m$} ;
			\foreach \x in {13,...,20}	\node at (\x-0.5,0.5) {$l$} ;
			
			\begin{scope}[thick,decoration=brace, yshift=-10pt]
				\draw[decorate] (5,0) -- 
					node[below=2pt]{\scriptsize left recursive call} (0,0);
				\draw[decorate] (11,0) -- 
					node[below=2pt]{\scriptsize middle recursive call} (6,0); 
				\draw[decorate] (20,0) -- 
					node[below=2pt]{\scriptsize right recursive call} (12,0);
			\end{scope}
		\end{scope}
		\def\y{-3cm}
		\begin{pgfonlayer}{background}
			\begin{scope}[black!5]
	 			\fill (3,0) -- (6,0) -- (8,\y) -- (5,\y) -- cycle ;
	 				 				 			\fill (15,0) -- (16,0) -- (12,\y) -- (11,\y) -- cycle ;
			\end{scope}
			\begin{scope}[black!50]
				\draw[densely dashed] (3,0) -- (5,\y)  (6,0) -- (8,\y) ;
				\draw[densely dotted] (6,0) -- (3,\y)  (8,0) -- (5,\y) ;
				\draw[densely dotted] (11,0) -- (15,\y)  (12,0) -- (16,\y) ;
				\draw[densely dashed] (15,0) -- (11,\y)  (16,0) -- (12,\y) ;
			\end{scope}
		\end{pgfonlayer}
	\end{tikzpicture}%
	}
	\caption{%
		\textbf{First row:} State of the array just after partitioning the ordinary
		elements (after \wref{lin:generalized-partition-call} of
		\wref{alg:generalized-yaroslavskiy}).
		The letters indicate whether the element at this location is smaller ($s$),
		between ($m$) or larger ($l$) than the two pivots $P$ and $Q$.
		Sample elements are {\setlength\fboxsep{2pt}\colorbox{black!10}{shaded}}.
		\protect\newline
		\textbf{Second row:} State of the array after pivots and sample parts have
		been moved to their partition (after \wref{lin:generalized-swap-2}).
		The “rubber bands” indicate moved regions of the array.
	}
	\label{fig:swapping-of-sampled-out}
\end{figure}

\oldparagraph{How do we deal with sampled-out elements?}
As discussed in \wref{sec:randomness-preservation}, we exclude sampled-out
elements from the partitioning range.
After partitioning, we thus have to move the $t_2$ sampled-out
elements, which actually belong between the pivots, to the middle
partition. 
Moreover, the pivots themselves have to be swapped in place. 
This process is illustrated in \wref{fig:swapping-of-sampled-out} and
spelled out in
lines~\ref*{lin:generalized-swap-t2-loop}\,--\,\ref*{lin:generalized-swap-2} of
\wref{alg:generalized-yaroslavskiy}.
Note that the order of swaps has been chosen carefully to correctly deal with
cases, where the regions to be exchanged overlap.

\subsection{Randomness Preservation}
\label{sec:randomness-preservation}

For analysis, it is vital to preserve the input distribution for
recursive calls, as this allows us to set up a recurrence equation for
costs, which in turn underlies the precise analysis of Quicksort.
While Yaroslavskiy's method (as given in \wref{alg:partition})
preserves randomness, pivot sampling requires special care.
For efficiently selecting the pivots, we \emph{sort} the entire sample, so the
sampled-out elements are far from randomly ordered; including
them in partitioning would not produce randomly ordered subarrays!
But there is also no need to include them in partitioning, as
we already have the sample divided into the three groups of 
$t_1$ small, $t_2$ medium and $t_3$ large elements.
All ordinary elements are still in random order and Yaroslavskiy's
partitioning divides them into three randomly ordered subarrays.

What remains problematic is the order of elements for recursive calls.
The second row in \wref{fig:swapping-of-sampled-out} shows the situation after
all sample elements (shaded gray) have been put into the correct subarray.
As the sample was sorted, the left and middle subarrays have sorted prefixes of
length $t_1$ resp.\ $t_2$ followed by a random permutation of the remaining
elements. Similarly, the right subarray has a sorted suffix of $t_3$ elements.
So the subarrays are \emph{not} randomly ordered, (except for the trivial case
$\vect t = 0$)! 
How shall we deal with this non-randomness?

The maybe surprising answer is that we can indeed \emph{exploit} this
non-randomness; not only in terms of a precise
analysis, but also for efficiency:
the sorted part \emph{always} lies completely inside the \emph{sample range} for
the next partitioning phase.
So our specific kind of non-randomness only affects sorting the sample (in
subsequent recursive calls), but it never affects the partitioning process
itself!

It seems natural that sorting should somehow be able to profit from partially
sorted input, and in fact, many sorting methods are known to be
\textsl{adaptive} to existing order \citep{EstivillCastro1992}.
For our special case of a fully sorted prefix or suffix of length $s\ge1$ and a
fully random rest, we can simply use Insertionsort where the first $s$
iterations of the outer loop are skipped.
Our Insertionsort implementations accept $s$ as an additional
parameter.
What is more, we can also precisely \emph{quantify} the savings resulting from
skipping the first $s$ iterations:
Apart from per-call overhead, we save exactly what it would have costed
to sort a random permutation of the length of this prefix/suffix
with Insertionsort. 
As all prefixes/suffixes have constant lengths (independent of the length
of the current subarray), precise analysis remains feasible.
Thereby, we need not be afraid of non-randomness \textit{per se}, 
as long as we can preserve the \emph{same kind} of non-randomness for recursive
calls and precisely analyze resulting costs.

\subsection{Generalized Yaroslavskiy Quicksort}

Combining the implementation of generalized pivot sampling\,---\,paying
attention to the subtleties discussed in the previous sections\,---\,with
Yaroslavskiy's partitioning method, 
we finally obtain \wref{alg:generalized-yaroslavskiy}.
We refer to this sorting method as \textsl{Generalized Yaroslavskiy Quicksort}
with pivot sampling parameter $\vect t = (t_1,t_2,t_3)$ and Insertionsort
threshold \isthreshold, shortly written as $\generalYarostM$.
We assume that $\isthreshold \ge k-1 = t_1+t_2+t_3+1$ to make sure
that every partitioning step has enough elements for pivot sampling.

The last parameter of \wref{alg:generalized-yaroslavskiy} tells the current call
whether it is a topmost call (\texttt{root}) or a recursive call on a left,
middle or right subarray of some earlier invocation. 
By that, we know which part of the array is already sorted:
for \texttt{root} calls, we cannot rely on anything being sorted,
in \texttt{left} and \texttt{middle} calls, we have a sorted prefix of length
$t_1$ resp.\ $t_2$, and for a \texttt{right} call, the $t_3$ rightmost
elements are known to be in order.
The initial call then takes the form
$\proc{GeneralizedYaroslavskiy}\,(\arrayA,1,n,\texttt{root})$.
\medskip

\begin{algorithm}[bhpt]
	\vspace{-1ex}
	\def\pl{\id{partLeft}}
	\def\pr{\id{partRight}}
	\def\pind{\id{i_p}}
	\def\qind{\id{i_q}}
	\def\case#1{$\kw{in case }\like[l]{\texttt{middle}\,}{\texttt{#1}}\;\kw{do}\;$}
		\begin{codebox}
		\Procname{$\proc{GeneralizedYaroslavskiy}\,(\arrayA,\id{left},\id{right},\id{type})$}
		\zi \Comment Assumes $\id{left} \le \id{right}$, $\isthreshold\ge k-1$
		\zi \Comment Sorts $A[\id{left},\ldots,\id{right}]$.
				\rule[-2ex]{0pt}{1ex}
		\li	\If $\id{right} - \id{left} < \isthreshold$
		\li	\Then
				\kw{case distinction} on $\id{type}$
		\li		\Do
		 			\case{root}		
		 			$\like[l]{\proc{InsertionSortRight}}{\proc{InsertionSortLeft}}
		 				\,(\arrayA,\id{left},\id{right},1)$
		\li			\case{left}	
					$\like[l]{\proc{InsertionSortRight}}{\proc{InsertionSortLeft}}
						\,(\arrayA,\id{left},\id{right},\max\{t_1,1\})$
		\li			\case{middle}	
					$\like[l]{\proc{InsertionSortRight}}{\proc{InsertionSortLeft}}
						\,(\arrayA,\id{left},\id{right},\max\{t_2,1\})$
					\li			\case{right}	
					$\proc{InsertionSortRight}\,(\arrayA,\id{left},\id{right},\max\{t_3,1\})$
				\End
		\li		\kw{end cases}
		\li	\Else
		\li			\kw{case distinction} on $\id{type}$ 
					\quad \Comment Sort sample
		\li			\Do
						\case{root}
						$\like[l]{\proc{SampleSortRight}}{\proc{SampleSortLeft}}
		 					\,(\arrayA,\id{left},\id{right},1)$
		\li				\case{left}
						$\like[l]{\proc{SampleSortRight}}{\proc{SampleSortLeft}}
		 					\,(\arrayA,\id{left},\id{right},\max\{t_1,1\})$
		\li				\case{middle}
						$\like[l]{\proc{SampleSortRight}}{\proc{SampleSortLeft}}
		 					\,(\arrayA,\id{left},\id{right},\max\{t_2,1\})$
		\li				\case{right}
						$\proc{SampleSortRight}
		 					\,(\arrayA,\id{left},\id{right},\max\{t_3,1\})$
					\End
		\li			\kw{end cases}
					\label{lin:generalized-sort-sample}
		\li			$p \gets \arrayA[\id{left} + t_1]$;
					\quad $q \gets \arrayA[\id{right}-t_3]$
		\li			$\pl \gets \id{left} + t_1 + t_2 + 1$;
					\quad $\pr \gets \id{right} - t_3 - 1$
		\li			$(\pind,\qind) \gets
						\proc{PartitionYaroslavskiy}\,(\arrayA,\pl,\pr,p,q)$
					\label{lin:generalized-partition-call}
					\rule[-2ex]{0pt}{1ex}
		\zi			\Comment Swap middle part of sample and $p$ to final place 
						(cf.\ \wref{fig:swapping-of-sampled-out})
		\li			\For $j \gets t_2 ,\ldots, 0$ 
					\quad \Comment iterate downwards
					\label{lin:generalized-swap-t2-loop}
		\li			\Do
						Swap $\arrayA[\id{left} + t_1 + j]$ and $\arrayA[\pind - t_2 + j]$
					\EndFor
		\zi			\Comment Swap $q$ to final place.
		\li			Swap $\arrayA[\qind]$ and $\arrayA[\pr+1]$
					\label{lin:generalized-swap-2}
					\rule[-2ex]{0pt}{1ex}
		\li			$\proc{GeneralizedYaroslavskiy}\,
						(\arrayA, 
						\like[l]{\pind-t_2+1,\;}{\id{left},}
						\like[l]{\pind-t_2-1,\;}{\pind-t_2-1,} 
						\like[l]{\texttt{middle}}{\texttt{left}})$
		\li			$\proc{GeneralizedYaroslavskiy}\,
						(\arrayA, 
						\like[l]{\pind-t_2+1,\;}{\pind-t_2+1,}
						\like[l]{\pind-t_2-1,\;}{\qind-1,} 
						\texttt{middle})$
		\li			$\proc{GeneralizedYaroslavskiy}\,
						(\arrayA,
						\like[l]{\pind-t_2+1,\;}{\qind+1,}
						\like[l]{\pind-t_2-1,\;}{\id{right},}
						\like[l]{\texttt{middle}}{\texttt{right}})$
			\EndIf
	\end{codebox}
	\vspace{-1ex}
	\caption{%
		\strut Yaroslavskiy's Dual-Pivot Quicksort with Generalized Pivot Sampling
	}
	\label{alg:generalized-yaroslavskiy}
\end{algorithm}

\begin{algorithm}[h]
	\vspace{-1ex}
	\def\pind{\id{i_p}}
	\def\qind{\id{i_q}}
	\begin{codebox}
		\Procname{$\proc{PartitionYaroslavskiy}\,(\arrayA,\id{left},\id{right},p,q)$}
		\zi \Comment Assumes $\id{left} \le \id{right}$. 
		\zi \Comment Rearranges \arrayA s.\,t.\ with return value $(\pind,\qind)$
				holds \smash{ $\begin{cases}
	 				\forall \: \id{left} \le j \le \pind,		& \arrayA[j] < p; \\
	 				\forall \: \id{\pind} < j < \qind,		& p \le \arrayA[j] \le q; \\
	 				\forall \: \qind \le j \le \id{right},		& \arrayA[j] \ge q .
				\end{cases}$}
				\rule[-2.5ex]{0pt}{1ex}
		\li $\ell\gets \id{left}$; 
		 	\quad $g\gets \id{right}$; 
		 	\quad $k\gets \ell$ \label{lin:yaroslavskiy-init-l-g-k} 
		\li	\While $k\le g$  \label{lin:yarosavskiy-outer-loop-branch}
		\li	\Do
				\If $\arrayA[k] < p$ \label{lin:yaroslavskiy-comp-1}
		\li		\Then
					Swap $\arrayA[k]$ and $\arrayA[\ell]$ \label{lin:yaroslavskiy-swap-1}
		\li			$\ell\gets \ell+1$ \label{lin:yaroslavskiy-l++-1}
		\li		\Else 
		\li			\If $\arrayA[k] \ge q$ \label{lin:yaroslavskiy-comp-2}
		\li			\Then
						\While $\arrayA[g] > q$ and $k<g$ \label{lin:yaroslavskiy-comp-3} 
		\li				\Do 
							$g\gets g-1$ 
						\EndWhile
		\li				\If $\arrayA[g] \ge p$ \label{lin:yaroslavskiy-comp-4}
		\li				\Then
							Swap $\arrayA[k]$ and $\arrayA[g]$ \label{lin:yaroslavskiy-swap-2}
		\li				\Else
		\li					Swap $\arrayA[k]$ and $\arrayA[g]$; \;
							Swap $\arrayA[k]$ and $\arrayA[\ell]$ \label{lin:yaroslavskiy-swap-3}
		\li					$\ell\gets \ell+1$ \label{lin:yaroslavskiy-l++-2}
						\EndIf
		\li				$g\gets g-1$ \label{lin:yaroslavskiy-g--}
					\EndIf
				\EndIf
		\li		$k\gets k+1$ \label{lin:yaroslavskiy-k++}
			\EndWhile \label{lin:yaroslavskiy-end-while}
		\li	$\ell\gets \ell-1$; \>\>\>$g\gets g+1$
		\li	\Return $(\ell, g)$
		\zi
	\end{codebox}
	\vspace{-4ex}
	\caption{\strut%
		Yaroslavskiy's dual-pivot partitioning algorithm.
	}
\label{alg:partition}
\end{algorithm}

\clearpage

\begin{algorithm}
	\vspace{-1ex}
	\begin{codebox}
		\Procname{$\proc{InsertionSortLeft}(\arrayA,\id{left},\id{right},s)$}
		\zi \Comment Assumes $\id{left} \le \id{right}$ and 
				$s \le \id{right}-\id{left} -1$.
		\zi \Comment Sorts $\arrayA[\id{left},\ldots,\id{right}]$, assuming that the
				$s$ \textit{leftmost} elements are already sorted.
				\rule[-2ex]{0pt}{1ex}
		\li	\For $i=\id{left} + s \,,\dots,\, \id{right}$
		\li	\Do
				$j\gets i-1$; \quad
				$v\gets \arrayA[i]$
		\li		\While $j \ge \id{left} \wbin\wedge v < \arrayA[j]$ 
					\label{lin:insertionsort-left-comp-1}
		\li		\Do
					$\arrayA[j+1] \gets \arrayA[j]$; \quad	\label{lin:insertionsort-left-write-1}
					$j\gets j-1$
				\EndWhile
		\li		$\arrayA[j+1] \gets v$ \label{lin:insertionsort-left-write-2}
			\EndFor
	\end{codebox}
	\vspace{-1ex}
	\caption{%
		\strut Insertionsort “from the left”, exploits sorted prefixes. 
	}
	\label{alg:insertionsort-left}
\end{algorithm}

\begin{algorithm}
	\vspace{-1ex}
	\begin{codebox}
		\Procname{$\proc{InsertionSortRight}(\arrayA,\id{left},\id{right},s)$}
		\zi \Comment Assumes $\id{left} \le \id{right}$ and 
				$s \le \id{right}-\id{left} -1$.
		\zi \Comment Sorts $\arrayA[\id{left},\ldots,\id{right}]$, assuming that the
				$s$ \textit{rightmost} elements are already sorted.
				\rule[-2ex]{0pt}{1ex}
		\li	\For $i=\id{right} - s \,,\dots,\, \id{left}$ 
			\quad\Comment iterate downwards 
		\li	\Do
				$j\gets i+1$; \quad
				$v\gets \arrayA[i]$
		\li		\While $j \le \id{right} \wbin\wedge v > \arrayA[j]$
						\label{lin:insertionsort-right-comp}
		\li		\Do
					$\arrayA[j-1] \gets \arrayA[j]$; \quad	\label{lin:insertionsort-right-write-1}
					$j\gets j+1$
				\EndWhile
		\li		$\arrayA[j-1] \gets v$ \label{lin:insertionsort-right-write-2}
			\EndFor
	\end{codebox}
	\vspace{-1ex}
	\caption{%
		\strut Insertionsort “from the right”, exploits sorted suffixes. 
	}
	\label{alg:insertionsort-right}
\end{algorithm}

\begin{algorithm}
	\vspace{-1ex}
	\begin{codebox}
		\Procname{$\proc{SampleSortLeft}(\arrayA,\id{left},\id{right},s)$}
		\zi \Comment Assumes $\id{right} - \id{left} + 1 \ge k$ and 
				$s \le t_1+t_2+1$.
		\zi \Comment Sorts the $k$ elements
		$\arrayA[\id{left}],\ldots,\arrayA[\id{left}+t_1+t_2],
		 \arrayA[\id{right}-t_3],\ldots,\arrayA[\id{right}]$, 
		\zi \Comment assuming that the $s$ leftmost 
				elements are already sorted.
				\rule[-1.75ex]{0pt}{1ex}
		\zi \Comment By $\arrayA\llbracket i \rrbracket$, 
			we denote the array cell $\arrayA[i]$, if $i \le \id{left} + t_1+t_2$, 
		\zi \Comment and $\arrayA[i+(n-k)]$ for $n=\id{right}-\id{left}+1$,
			otherwise.
				\rule[-2ex]{0pt}{1ex}
		\li $\proc{InsertionSortLeft}(\arrayA,\id{left},
					\id{left}+t_1+t_2,s)$
		\li	\For $i=\id{left} + t_1+t_2+1 \,,\dots,\, \id{left} + k - 1$
		\li	\Do
				$j\gets i-1$; \quad
				$v\gets \arrayA \llbracket i \rrbracket$
		\li		\While $j \ge \id{left} \wbin\wedge v < \arrayA\llbracket j\rrbracket$
		\li		\Do
					$\arrayA\llbracket j+1\rrbracket \gets \arrayA\llbracket j\rrbracket$;
					\quad \label{samplesort-left-write-1} $j\gets j-1$
				\EndWhile
		\li		$\arrayA\llbracket j+1\rrbracket \gets v$
				\label{lin:samplesort-left-write-2}
			\EndFor
	\end{codebox}
	\vspace{-1ex}
	
	\caption{%
		\strut Sorts the sample with Insertionsort “from the left” 
	}
	\label{alg:samplesort-left}
\end{algorithm}

\begin{algorithm}
	\vspace{-1ex}
	\begin{codebox}
		\Procname{$\proc{SampleSortRight}(\arrayA,\id{left},\id{right},s)$}
		\zi \Comment Assumes $\id{right} - \id{left} + 1 \ge k$ and 
				$s \le t_3+1$.
		\zi \Comment Sorts the $k$ elements
		$\arrayA[\id{left}],\ldots,\arrayA[\id{left}+t_1+t_2],
		 \arrayA[\id{right}-t_3],\ldots,\arrayA[\id{right}]$, 
		\zi \Comment assuming that the $s$ rightmost elements are
			already sorted.
				\rule[-1.75ex]{0pt}{1ex}
		\zi \Comment By $\arrayA\llbracket i \rrbracket$, 
			we denote the array cell $\arrayA[i]$, if $i \le \id{left} + t_1+t_2$, 
		\zi \Comment and $\arrayA[i+(n-k)]$ for $n=\id{right}-\id{left}+1$,
			otherwise.
				\rule[-2ex]{0pt}{1ex}
		\li $\proc{InsertionSortRight}(\arrayA,
					\id{right}-t_3,\id{right},s)$
		\li	\For $i=\id{left} + k - t_3 - 2 \,,\dots,\, \id{left}$ 
			\quad\Comment iterate downwards 
		\li	\Do
				$j\gets i+1$; \quad
				$v\gets \arrayA \llbracket i\rrbracket$
		\li		\While $j \le \id{left} + k \wbin\wedge 
					v > \arrayA \llbracket j \rrbracket$ 
		\li		\Do
					$\arrayA \llbracket j-1 \rrbracket \gets 
					\arrayA \llbracket j \rrbracket$; \quad	
					\label{lin:samplesort-right-write-1}
					$j\gets j+1$
				\EndWhile
		\li		$\arrayA \llbracket j-1 \rrbracket \gets v$
		\label{lin:samplesort-right-write-2}
			\EndFor
	\end{codebox}
	\vspace{-1ex}
	\caption{%
		\strut Sorts the sample with Insertionsort “from the right”
	}
	\label{alg:samplesort-right}
\end{algorithm}

\FloatBarrier

\clearpage
\section{Properties of Distributions}
\label{app:distributions}

We herein collect definitions and basic properties of the distributions used in
this paper. 
They will be needed for computing expected values in
\wref{app:proof-of-lem-expectations}.
We use the notation $x^{\overline n}$ and $x^{\underline n}$ of
\citet{ConcreteMathematics} for rising and falling factorial powers, respectively.

\subsection{Dirichlet Distribution and Beta Function}
\label{sec:dirichlet-dist-beta-function}
For $d\in\N$ let $\Delta_d$ be the standard $(d-1)$-dimensional simplex, \ie, 
\begin{align}
\label{eq:def-delta-d}
		\Delta_d
	&\wwrel\ce 
		\biggl\{
			x = (x_1,\ldots,x_d) 
			\wrel: 
			\forall i : x_i \ge 0 \; 
			\rel\wedge 
			\sum_{\mathclap{1\le i \le d}} x_i = 1
		\biggr\} \;.
\end{align}
Let $\alpha_1,\ldots,\alpha_d > 0$ be positive reals.
A random variable $\vect X \in \R^d$ is said to have the 
\emph{Dirichlet distribution} with \emph{shape parameter} 
$\vect\alpha \ce (\alpha_1,\ldots,\alpha_d)$\,---\,abbreviated as
$\vect X \eqdist \dirichlet(\vect\alpha)$\,---\,if it has a density given by
\begin{align}
\label{eq:def-dirichlet-density}
		f_{\vect X}(x_1,\ldots,x_d)
	&\wwrel\ce \begin{cases}
			\frac1{\BetaFun(\vect\alpha)} \cdot
			x_1^{\alpha_1 - 1} \cdots x_d^{\alpha_d-1} ,
			& \text{if } \vect x \in \Delta_d \,; \\
			0 , & \text{otherwise} \..
		\end{cases}
\end{align}
Here, $\BetaFun(\vect\alpha)$ is the \emph{$d$-dimensional Beta function}
defined as the following Lebesgue integral:
\begin{align}
\label{eq:def-beta-function}
		\BetaFun(\alpha_1,\ldots,\alpha_d)
	&\wwrel\ce
		\int_{\Delta_d} x_1^{\alpha_1 - 1} \cdots x_d^{\alpha_d-1} \; \mu(d \vect x)
	\;.
\end{align}
The integrand is exactly the density without the normalization constant
$\frac1{\BetaFun(\alpha)}$, hence $\int f_X \,d\mu = 1$ as needed for
probability distributions.

The Beta function can be written in terms of the Gamma function 
$\Gamma(t) = \int_0^\infty x^{t-1} e^{-x} \,dx$
as
\begin{align}
\label{eq:beta-function-via-gamma}
		\BetaFun(\alpha_1,\ldots,\alpha_d)
	&\wwrel=
		\frac{\Gamma(\alpha_1) \cdots \Gamma(\alpha_d)}
		{\Gamma(\alpha_1+\cdots+\alpha_d)} \;.
\end{align}
(For integral parameters $\vect\alpha$, a simple inductive argument and partial
integration suffice to prove~\wref{eq:beta-function-via-gamma}.)\\
Note that $\dirichlet(1,\ldots,1)$ corresponds to the uniform distribution over
$\Delta_d$.
For integral parameters $\vect\alpha\in\N^d$, $\dirichlet(\vect\alpha)$ is the
distribution of the \emph{spacings} or \emph{consecutive differences} induced by
appropriate order statistics of i.\,i.\,d.\ uniformly in $(0,1)$ distributed
random variables: 
\begin{proposition}[{{\citealt[Section\,6.4]{David2003}}}]
\label{pro:spacings-dirichlet-general-dimension}
	Let $\vect\alpha \in \N^d$ be a vector of positive integers and set $k \ce -1 +
	\sum_{i=1}^d \alpha_i$. Further let $V_1,\ldots,V_{k}$ be $k$ random variables
	i.\,i.\,d.\ uniformly in $(0,1)$ distributed.
	Denote by \smash{$V_{(1)}\le \cdots \le V_{(k)}$} their corresponding
	order statistics.
	We select some of the order statistics according to $\vect \alpha$: 
	for $j=1,\ldots,d-1$ define \smash{$W_j \ce V_{(p_j)}$}, where $p_j \ce
	\sum_{i=1}^j \alpha_i$. Additionally, we set $W_0 \ce 0$ and $W_d \ce 1$.
	
	Then, the \textit{consecutive distances} (or \textit{spacings}) $D_j \ce W_j -
	W_{j-1}$ for $j=1,\ldots,d$ induced by the selected
	order statistics $W_1,\ldots,W_{d-1}$ are Dirichlet
	distributed with parameter $\vect \alpha$: 
	\begin{align*}
			(D_1,\ldots,D_d) 
		&\wwrel\eqdist 
			\dirichlet(\alpha_1,\ldots,\alpha_d) \;.
	\end{align*}%
\qed\end{proposition}
\smallskip

In the computations of \wref{sec:expectations}, mixed moments of Dirichlet distributed
variables will show up, which can be dealt with using the following general
statement.
\begin{lemma}
\label{lem:dirichlet-mixed-moments}
	Let $\vect X = (X_1,\ldots,X_d) \in \R^d$ be a $\dirichlet(\vect\alpha)$
	distributed random variable with parameter $\vect\alpha = (\alpha_1,\ldots,\alpha_d)$.
	Let further $m_1,\ldots,m_d \in \N$ be non-negative integers and abbreviate
	the sums $A \ce \sum_{i=1}^d \alpha_i$ and $M \ce \sum_{i=1}^d m_i$. 
	Then we	have 
	\begin{align*}
			\E\bigl[ X_1^{m_1} \cdots X_d^{m_d} \bigr]
		&\wwrel=
			\frac{\alpha_1^{\overline{m_1}} \cdots \alpha_d^{\overline{m_d}}}
				{A^{\overline M}} \;.
	\end{align*}
\end{lemma}

\begin{proof}
Using $\frac{\Gamma(z+n)}{\Gamma(z)} = z^{\overline n}$ for all
$z\in\R_{>0}$ and $n \in \N$, we compute
\begin{align}
		\E\bigl[ X_1^{m_1} \cdots X_d^{m_d} \bigr]
	&\wwrel=
		\int_{\Delta_d} 
			x_1^{m_1} \cdots x_d^{m_d} \cdot
			\frac{x_1^{\alpha_1-1} \cdots x_d^{\alpha_d-1}}{\BetaFun(\vect\alpha)}
		\; \mu(dx)
	\\	&\wwrel=
		\frac{\BetaFun(\alpha_1 + m_1,\ldots,\alpha_d + m_d)}
			{\BetaFun(\alpha_1,\ldots,\alpha_d)}
	\\	&\wwrel{\eqwithref{eq:beta-function-via-gamma}}
		\frac{ \alpha_1^{\overline{m_1}} \cdots \alpha_d^{\overline{m_d}} }
			{ A^{\overline M} } \;.
\end{align}
\end{proof}

For completeness, we state here a two-dimensional Beta integral with an
additional logarithmic factor that is needed in \wref{app:CMT-solution} (see
also \citealt[Appendix~B]{Martinez2001}):
\begin{align}
		\BetaFun_{\ln}(\alpha_1,\alpha_2) 
	&\wwrel\ce 
		- \int_0^1 x^{\alpha_1-1} (1-x)^{\alpha_2-1} \ln x \, dx
\notag\\	&\wwrel{\like[r]\ce=}
		\BetaFun(\alpha_1, \alpha_2) 
		(\harm{\alpha_1+\alpha_2-1} -	\harm{\alpha_1-1}) \;.
\label{eq:beta-log}
\end{align}

For integral parameters $\vect\alpha$, the proof is elementary: 
By partial integration, we can find a recurrence equation for $\BetaFun_{\ln}$:
\begin{align*}
		\BetaFun_{\ln}(\alpha_1,\alpha_2)
	&\wwrel=
		\frac1{\alpha_1} \BetaFun(\alpha_1,\alpha_2) \bin+ 
		\frac{\alpha_2-1}{\alpha_1} \BetaFun_{\ln}(\alpha_1+1,\alpha_2-1) \;.
\end{align*}
Iterating this recurrence until we reach the base case $\BetaFun_{\ln}(a,0) =
\frac1{a^2}$ and using \wref{eq:beta-function-via-gamma} to expand the
Beta function, we obtain \wref{eq:beta-log}.

\needspace{5cm}
\subsection{Multinomial Distribution}
\label{sec:multinomial-distribution}

Let $n,d \in \N$ and $k_1,\ldots,k_d\in\N$. \emph{Multinomial coefficients} are
a multidimensional extension of binomials:
\begin{align*}
		\binom{n}{k_1,k_2,\ldots,k_d}
	&\wwrel\ce 
	\begin{cases} \displaystyle
			\frac{n!}{k_1! k_2! \cdots k_d!} ,
		& \displaystyle \text{if } 
			n=\sum_{i=1}^d k_i \;; \\[1ex]
		0 , & \text{otherwise} .
	\end{cases}
\end{align*}
Combinatorially, $\binom{n}{k_1,\ldots,k_d}$ is the number of ways to
partition a set of $n$ objects into $d$ subsets of respective sizes
$k_1,\ldots,k_d$ and thus 
they appear naturally in the \weakemph{multinomial theorem}:
\begin{align}
\label{eq:multinomial-theorem}
		(x_1 + \cdots + x_d)^n
	&\wwrel= \mkern-10mu
		\sum_{\substack{i_1,\ldots,i_d \in \N \\ i_1+\cdots+i_d = n}} \mkern-5mu 
			\binom{n}{i_1,\ldots,i_d} \; x_1^{i_1} \cdots x_d^{i_d} 
		\qquad\qquad\text{for } n\in\N \;.
\end{align}

Let $p_1,\ldots,p_d \in [0,1]$ such that $\sum_{i=1}^d p_i = 1$.
A random variable $\vect X\in\N^d$ is said to have \emph{multinomial
distribution} with parameters $n$ and $\vect p = (p_1,\ldots,p_d)$\,---\,written shortly
as~$\vect X \eqdist \multinomial(n,\vect p)$\,---\,if for any 
$\vect i = (i_1,\ldots,i_d) \in \N^d$
holds
\begin{align*}
		\Prob(\vect X = \vect i)
	&\wwrel=
		\binom{n}{i_1,\ldots,i_d} \; p_1^{i_1} \cdots p_d^{i_d} \;.
\end{align*}

We need some expected values involving multinomial variables.
They can be expressed as special cases of the following mixed factorial moments.

\begin{lemma}
	\label{lem:multinomial-mixed-factorial-moments}
	Let $p_1,\ldots,p_d \in [0,1]$ such that $\sum_{i=1}^d p_i =1$
	and consider a $\multinomial(n,\vect p)$ distributed variable
	$\vect X = (X_1,\ldots,X_d) \in \N^d$.
	Let further $m_1,\ldots,m_d \in \N$ be non-negative integers and abbreviate 
	their sum as $M \ce \sum_{i=1}^d m_i$. 
	Then we have 
	\begin{align*} 
			\E\bigl[
				(X_1)^{\underline{m_1}} \cdots (X_d)^{\underline{m_d}} 
			\bigr] 
		&\wwrel=
			n^{\underline M} \, p_1^{m_1} \cdots p_d^{m_d} \;.
	\end{align*}
\end{lemma}

\begin{proof}
We compute
\begin{align}
		\E\bigl[ (X_1)^{\underline{m_1}} \cdots (X_d)^{\underline{m_d}} \bigr]
	&\wwrel=
		\sum_{\vect x \in \N^d} 
			x_1^{\,\underline{m_1}}\cdots x_d^{\,\underline{m_d}}
			\binom{n}{x_1,\ldots,x_d} \;
			p_1^{x_1} \cdots p_d^{x_d} 
	\nonumber\\ &\wwrel=
		n^{\underline M} \, p_1^{m_1} \cdots p_d^{m_d} \times{} 
	\nonumber\\*&\wwrel\ppe
		\sum_{\substack{\vect x \in \N^d : \\ \forall i : x_i \ge m_i}}
			\mkern -10mu \binom{n-M}{x_1-m_1,\ldots,x_d-m_d} \;
		p_1^{x_1-m_1} \cdots p_d^{x_d-m_d} 
	\nonumber\\ &\wwrel{\eqwithref{eq:multinomial-theorem}}
		n^{\underline M} \, p_1^{m_1} \cdots p_d^{m_d}
		\;  
		\bigl(\.\. \smash{ \underbrace{p_1 + \cdots + p_d} _ {=1} } \.\. \bigr)^{n-M}
	\nonumber\\ &\wwrel=
		n^{\underline M} \, p_1^{m_1} \cdots p_d^{m_d}
		\;.
\end{align}
\end{proof}

\clearpage

\section[Proof of Lemma 6.1]{Proof of \wref{lem:expectations}}
\label{app:proof-of-lem-expectations}

We recall that
$		\vect D
	\eqdist	
		\dirichlet(\vect t + 1)
$ and 
$
		\vect I
	\eqdist
		\multinomial(n-k,\vect D)
$
and start with the simple ingredients: $\E[I_j]$ for $j=1,2,3$.
\begin{align}
		\E[I_j]	
	&\wwrel=	
		\E_{\vect D} \bigl[ \E[ I_j \given \vect D = \vect d] \bigr]
\nonumber\\	&\wwrel{\eqwithref[r]{lem:multinomial-mixed-factorial-moments}}	
							\E_{\vect D} \bigl[ D_j (n-k) \bigr] 
\nonumber\\	&\wwrel{\eqwithref[r]{lem:dirichlet-mixed-moments}}	
							(n-k) \frac{t_j+1}{k+1}
		\;.
\label{eq:expectation-Ij} 
\end{align}
The term $\E\bigl[\bernoulli\bigl(\frac{I_3}{n-k}\bigr)\bigr]$ is then easily
computed using~\wref{eq:expectation-Ij}:
\begin{align}
\label{eq:expectation-delta}
		\E\bigl[\bernoulli\bigl(\tfrac{I_3}{n-k}\bigr)\bigr]
	&\wwrel=
		\frac{\E[{I_3}]}{n-k}
	\wwrel=
		\frac{t_3+1}{k+1}
	\wwrel{\wwrel=} \Theta(1) \;.
\end{align}
This leaves us with the hypergeometric variables; using the well-known formula
$\E[\hypergeometric(k,r,n)] = k\frac rn$, we find
\begin{align}
		\E\bigl[ \hypergeometric(I_1+I_2, I_3, n-k) \bigr]
	&\wwrel=
		\E_{\vect I} \Bigl[ 
			\E\bigl[ \hypergeometric(i_1+i_2, i_3, n-k) \given \vect I = \vect i \bigr] 
		\Bigr]
	\nonumber\\	&\wwrel= 
		\E\left[ \frac{(I_1+I_2) I_3}{n-k} \right]
	\nonumber\\	&\wwrel=
		\E_{\vect D} \left[ 
			  \frac{ \E[ I_1 I_3 \given \vect D] 
			+ \E[ I_2 I_3 \given \vect D ] }{n-k}
		\right]
	\nonumber\\ &\wwrel{\eqwithref[r]{lem:multinomial-mixed-factorial-moments}}
		\frac{ (n-k)^{\underline 2} \E[D_1 D_3] + 
				(n-k)^{\underline 2} \E[D_2 D_3] } 
			{n-k}
	\nonumber\\ &\wwrel{\eqwithref[r]{lem:dirichlet-mixed-moments}}
		\frac{\bigl((t_1+1)+(t_2+1)\bigr)(t_3+1)}{(k+1)^{\overline2}} (n-k-1)
		\;.
\label{eq:expectation-sm-at-G}  
\end{align}
The second hypergeometric summand is obtained similarly. 
\hfill\proofSymbol

\clearpage
\section{Solution to the Recurrence}
\label{app:CMT-solution}

An elementary proof can be given for
\wref{thm:leading-term-expectation-hennequin} using
\citeauthor{Roura2001}'s \emph{Continuous Master Theorem} (CMT)
\citep{Roura2001}.
The CMT applies to a wide class of full-history recurrences whose coefficients
can be well-approximated asymptotically by a so-called \emph{shape function}
$w:[0,1] \to \R$. 
The shape function describes the coefficients only depending on the \emph{ratio}
$j/n$ of the subproblem size $j$ and the current size $n$ (not depending on $n$
or $j$ itself) and it smoothly continues their behavior to any real number
$z\in[0,1]$.
This continuous point of view also allows to compute precise asymptotics
for complex discrete recurrences via fairly simple integrals.

\begin{theorem}[{{\citealt[Theorem~18]{Martinez2001}}}]
\label{thm:CMT}
	Let $F_n$ be recursively defined~by
	\begin{align}
	\label{eq:CMT-recurrence}
		F_n \wwrel= \begin{cases}
			b_n,	&\text{for~} 0 \le n < N; \\
			\displaystyle{ \vphantom{\bigg|}
				t_n \bin+ \smash{\sum_{j=0}^{n-1} w_{n,j} \, F_j}, 
			} 	&\text{for~} n \ge N\,
		\end{cases}
	\end{align}
	where the toll function satisfies $t_n \sim K n^\alpha \log^\beta(n)$ as
	$n\to\infty$ for constants $K\ne0$, $\alpha\ge0$ and $\beta > -1$.
	Assume there exists a function $w:[0,1]\to \R$, 
		such that 
	\begin{align}
	\label{eq:CMT-shape-function-condition}
		\sum_{j=0}^{n-1} \,\biggl|
			w_{n,j} \bin- \! \int_{j/n}^{(j+1)/n} \mkern-15mu w(z) \: dz
		\biggr|
		\wwrel= \Oh(n^{-d}),
		\qquad\qquad(n\to\infty),
	\end{align}
	for a constant $d>0$.
	With \smash{$\displaystyle H \ce 1 - \int_0^1 \!z^\alpha w(z) \, dz$}, we
	have the following cases:
	\begin{enumerate}[itemsep=0ex]
		\item If $H > 0$, then $\displaystyle F_n \sim \frac{t_n}{H}$.
		\item \label{case:CMT-H0} 
		If $H = 0$, then 
		$\displaystyle
		F_n \sim \frac{t_n \ln n}{\tilde H}$ with 
		$\displaystyle \tilde H = -(\beta+1)\int_0^1 \!z^\alpha \ln(z) \, w(z) \, dz$.
		\item If $H < 0$, then $F_n \sim \Theta(n^c)$ for the unique
		$c\in\R$ with $\displaystyle\int_0^1 \!z^c w(z) \, dz = 1$.
	\end{enumerate}
\qed\end{theorem}

\smallskip\noindent
The analysis of single-pivot Quicksort with pivot sampling is the application
par excellence for the CMT \citep{Martinez2001}. 
We will generalize this work of \citeauthor{Martinez2001} to the dual pivot
case.

\subsection{Rewriting the Recurrence}
We start from the distributional
equation~\wref{eq:distributional-recurrence} by conditioning 
on $\vect J$. For $n > \isthreshold$, this gives
\begin{align*}
		C_n 
	&\wwrel= 
		T_n \wbin+ 
		\sum_{j=0}^{n-2} \Bigl( 
			\indicator{J_1=j} C_j
			\bin+\indicator{J_2=j} C'_j
			\bin+\indicator{J_3=j} C''_j
		\Bigr)
		\;.
\end{align*}
Taking expectations on both sides and exploiting independence yields
\begin{align*}
		\E[C_n] 
	&\wwrel= 
		\E[T_n] \wbin+ \sum_{l=1}^3 \sum_{j=0}^{n-2} \E[\indicator{J_l=j}] \E[C_j] 
\\	&\wwrel=
		\E[T_n] \wbin+ \sum_{j=0}^{n-2} \bigl(
			\Prob(J_1=j) + \Prob(J_2=j) + \Prob(J_3=j)
		\bigr) \E[C_j]\,,
\end{align*}
which is a recurrence in the form of \wref{eq:CMT-recurrence} with weights
\begin{align*}
		w_{n,j} 
	&\wwrel= 
		\Prob(J_1=j) \bin+ \Prob(J_2=j) \bin+ \Prob(J_3=j) \;.
\end{align*}
(Note that the probabilities implicitly depend on $n$.) \\
By definition, $\Prob(J_l=j) = \Prob(I_l = j-t_l)$ for $l=1,2,3$.
The latter probabilities can be computed using that the marginal distribution of
$I_l$ is binomial $\binomial(N,D_l)$, where we abbreviate by $N \ce n-k$ the
number of ordinary elements. 
It is convenient to consider $\vect{\tilde D} \ce (D_l,1-D_l)$,
which is distributed like
$\vect{\tilde D} \eqdistt \dirichlet(t_l+1,k-t_l)$.
For $i\in[0..N]$ holds
\begin{align}
		\Prob(I_l=i)
	&\wwrel=
		\E_{\vect D}\bigl[ \E_{\vect J}[ \indicator{I_l=i} \given \vect D ] \bigr]
\notag\\	&\wwrel=
		\E_{\vect D}\bigl[
			\tbinom Ni \tilde D_1^i \tilde D_2^{N-i} 
		\bigr]
\notag\\	&\wwrel{\eqwithref[r]{lem:dirichlet-mixed-moments}}
		\binom Ni 
		\frac{(t_l+1)^{\overline i}(k-t_l)^{\overline{N-i}}}
			{(k+1)^{\overline N}}\;.
\label{eq:prob-Il-equals-i}
\end{align}

\subsection{Finding a Shape Function}
In general, a good guess for the shape function is $w(z) = \lim_{n\to\infty}
n\,w_{n,zn}$ \citep{Roura2001} and, indeed, this will work out for our
weights.
We start by considering the behavior for large $n$ of the terms 
$\Prob(I_l = zn + r)$ for $l=1,2,3$, where $r$ does not depend on $n$.
Assuming $zn+r \in \{0,\ldots,n\}$, we compute
\begin{align}
		\Prob(I_l=zn+r)
	&\wwrel= 
		\binom N{zn+r} 
		\frac{(t_l+1)^{\overline{zn+r}}(k-t_l)^{\overline{(1-z)n-r}}}
			{(k+1)^{\overline N}}
\notag\\	&\wwrel=
		\frac{N!}{(zn+r)!((1-z)n-r)!} 
		\frac{\displaystyle \frac{(zn + r + t_l)!} {t_l!}  \,
				\frac{\bigl((1-z)n - r + k - t_l - 1\bigr)!} {(k-t_l-1)!}}
		{\displaystyle \frac{(k+N)!}{k!}}
\notag\\	&\wwrel=
		\frac{k!}{t_l!(k-t_l-1)!}
		\frac{(zn+r+t_l)^{\underline{t_l}}  \,
				\bigl((1-z)n + - r + k - t_l + 1\bigr)^{\underline{k-t_l-1}}}
		{n^{\underline k}}
		\,,
\intertext{%
	and since this is a \emph{rational} function in $n$\,,
}
	&\wwrel=
		(k-t_l)\binom{k}{t_l}  
		\frac{ (zn)^{t_l} ((1-z)n)^{k-t_l-1}}{n^k}
		\cdot \Bigl(
			1 \bin+ \Oh(n^{-1})
		\Bigr)
\notag\\ 	&\wwrel=
		\underbrace{
			(k-t_l)\binom{k}{t_l} z^{t_l} (1-z)^{k-t_l-1}
		} _ {\equalscolon w_l(z)}
		\cdot \Bigl(
			n^{-1} \bin+ \Oh(n^{-2})
		\Bigr)\,,
		\qquad (n\to\infty).
\label{eq:CMT-n-prob-Il-zn-limit}
\end{align}
Thus 
$n\.\Prob(J_l = zn)  \rel=  n\.\Prob(I_l = zn-t_l)  \rel\sim  w_l(z)$, and
our candidate for the shape function is
\begin{align*}
		w(z)
	&\wwrel=
		\sum_{l=1}^3 w_l(z) 
	\wwrel= 
		\sum_{l=1}^3 (k-t_l)\binom{k}{t_l} z^{t_l} (1-z)^{k-t_l-1}\;.
\end{align*}
It remains to verify condition \wref{eq:CMT-shape-function-condition}.
We first note using \wref{eq:CMT-n-prob-Il-zn-limit} that
\begin{align}
		n \. w_{n,zn} 
	&\wwrel=
		w(z) \bin+ \Oh(n^{-1})
	\;.
\label{eq:CMT-w-n-zn-asymptotic}
\end{align}
Furthermore as $w(z)$ is a \emph{polynomial} in $z$, its derivative exists and
is finite in the compact interval $[0,1]$, so its absolute value is bounded by
a constant~$C_w$.
Thus $w:[0,1]\to\R$ is \textsl{Lipschitz-continuous} with Lipschitz constant
$C_w$:
\begin{align}
\label{eq:CMT-wz-Lipschitz}
	\forall z,z'\in[0,1] 
	&\wwrel:
		\bigl|w(z) - w(z')\bigr| 
		\wrel\le 
		C_w |z-z'|
		\;.
\end{align}
For the integral from \wref{eq:CMT-shape-function-condition}, we then have
\begin{align*}
		\sum_{j=0}^{n-1} \,\biggl|
			w_{n,j} \bin- \! \int_{j/n}^{(j+1)/n} \mkern-15mu w(z) \: dz
		\biggr|
	&\wwrel=
		\sum_{j=0}^{n-1} \,\biggl|
			\int_{j/n}^{(j+1)/n} \mkern-15mu n \. w_{n,j} - w(z) \: dz
		\biggr|
\\	&\wwrel\le
		\sum_{j=0}^{n-1} \frac1n \cdot 
			\max_{z\in \bigl[\frac jn, \frac{j+1}n \bigr]}
				\Bigl| n \. w_{n,j} - w(z) \Bigr|
\\	&\wwrel{\eqwithref{eq:CMT-w-n-zn-asymptotic}}
		\sum_{j=0}^{n-1} \frac1n \cdot
			\Biggl[
				\max_{\;z\in \bigl[\frac jn, \frac{j+1}n \bigr]}
				\Bigl| w(j/n) - w(z)\Bigr|  \wbin+ \Oh(n^{-1})
			\Biggr]
\\	&\wwrel\le
		\Oh(n^{-1}) \wbin+
			\max_{\substack{z,z'\in[0,1]:\\ |z-z'|\le 1/n}}
			\bigl| w(z) - w(z')\bigr|
\\	&\wwrel{\relwithref{eq:CMT-wz-Lipschitz}{\le}}
		\Oh(n^{-1}) \wbin+
		C_w \frac1n
\\	&\wwrel= \Oh (n^{-1}) \,,
\end{align*}
which shows that our $w(z)$ is indeed a shape function of our recurrence
(with $d=1$).

\subsection{Applying the CMT}
With the shape function $w(z)$ we can apply \wref{thm:CMT} with $\alpha=1$,
$\beta=0$ and $K=a$.
It turns out that \wildref[case]{case:CMT-H0}{\ref*} of the CMT applies:
\begin{align*}
		H
	&\wwrel=
		1 \bin- \int_0^1 z \, w(z) \, dz
\\	&\wwrel=
		1 \bin-\sum_{l=1}^3 \int_0^1 z \, w_l(z) \, dz
\\	&\wwrel=
		1 \bin-\sum_{l=1}^3 (k-t_l)\binom{k}{t_l} \BetaFun(t_l+2,k-t_l)
\\	&\wwrel{\eqwithref{eq:beta-function-via-gamma}} 
		1 \bin - \sum_{l=1}^3\frac{t_l+1}{k+1}
	\wwrel= 0\;.
\end{align*}  
For this case, the leading term coefficient of the solution is $t_n \ln(n) /
\tilde H$ with
\begin{align*}
		\tilde H
	&\wwrel=
		- \int_0^1 z \ln(z) \, w(z) \, dz
\\	&\wwrel=
		\sum_{l=1}^3 (k-t_l)\binom{k}{t_l} \BetaFun_{\ln}(t_l+2,k-t_l)
\\	&\wwrel{\eqwithref{eq:beta-log}}
		\sum_{l=1}^3 (k-t_l)\binom{k}{t_l} 
			\BetaFun(t_l+2,k-t_l)(\harm{k+1} - \harm{t_l+1})
\\	&\wwrel=
		\sum_{l=1}^3 \frac{t_l+1}{k+1}(\harm{k+1} - \harm{t_l+1})\;.
\end{align*}
So indeed, we find $\tilde H = \discreteEntropy$ as claimed in
\wref{thm:leading-term-expectation-hennequin}, concluding the proof.

Note that the above arguments actually \emph{derive}\,---\,not only prove
correctness of\,---\,the precise leading term asymptotics of a quite involved 
recurrence equation. 
Compared with \citeauthor{hennequin1991analyse}'s original proof via generating
functions, it needed much less mathematical theory.

\end{document}